\documentclass[twocolumn, prx, superscriptaddress,notitlepage]{revtex4-2}
\usepackage{graphicx}
\usepackage{dcolumn}
\usepackage{bm}
\usepackage[usenames,dvipsnames]{color}
\usepackage[most]{tcolorbox}
\usepackage{multirow}
\usepackage{gensymb}
\usepackage[normalem]{ulem}
\usepackage{CJK}
\usepackage{tablefootnote}
\usepackage{footnote}
\usepackage[para]{footmisc}
\usepackage{comment}
\usepackage[colorlinks, linkcolor=blue,anchorcolor=blue,citecolor=blue,urlcolor=blue]{hyperref}
\usepackage{pifont}
\usepackage{physics}
\usepackage{natbib}
\usepackage{mathtools}
\usepackage{amsmath}
\usepackage{amsfonts}
\usepackage{amssymb}
\usepackage{amsthm}

\newtheorem{theorem}{Theorem}

\usepackage{algorithm}
\usepackage{algpseudocode}
\usepackage{qcircuit}

\usepackage{booktabs}
\usepackage{colortbl}

\usepackage{color,soul}
\usepackage[mathscr]{euscript}
\begin{document}
\begin{CJK*}{UTF8}{}

\title{Privacy-preserving quantum federated learning via gradient hiding}

\author{Changhao Li}
\email{changhao.li@jpmchase.com}
\affiliation{Global Technology Applied Research, JPMorgan Chase, New York, NY 10017 USA}

\author{Niraj Kumar}
\affiliation{Global Technology Applied Research, JPMorgan Chase, New York, NY 10017 USA}

\author{Zhixin Song}
\affiliation{Global Technology Applied Research, JPMorgan Chase, New York, NY 10017 USA}

\author{Shouvanik Chakrabarti}
\affiliation{Global Technology Applied Research, JPMorgan Chase, New York, NY 10017 USA}

\author{Marco Pistoia}
\affiliation{Global Technology Applied Research, JPMorgan Chase, New York, NY 10017 USA}

\begin{abstract}
Distributed quantum computing, particularly distributed quantum machine learning, has gained substantial prominence for its capacity to harness the collective power of distributed quantum resources, transcending the limitations of individual quantum nodes. Meanwhile, the critical concern of privacy within distributed computing protocols remains a significant challenge, particularly in standard classical federated learning (FL) scenarios where data of participating clients is susceptible to leakage via gradient inversion attacks by the server. This paper presents innovative quantum protocols with quantum communication designed to address the FL problem, strengthen privacy measures, and optimize communication efficiency. In contrast to previous works that leverage expressive variational quantum circuits or differential privacy techniques, we consider gradient information concealment using quantum states and propose two distinct FL protocols, one based on private inner-product estimation and the other on incremental learning. These protocols offer substantial advancements in privacy preservation with low communication resources, forging a path toward efficient quantum communication-assisted FL protocols and contributing to the development of secure distributed quantum machine learning, thus addressing critical privacy concerns in the quantum computing era.
\end{abstract}

\maketitle
\end{CJK*}

\section{Introduction}

Quantum computing has experienced rapid advancements in recent years, and within this dynamic landscape, distributed quantum computing including quantum machine learning (QML)~\cite{Cuomo2020DQC,Caleffi2022distributed,Beals2013,Cacciapuoti2020,montanaro2023quantum,gilboa2023exponential,Li2023blindQML,PhysRevLett.130.150602, kumar2023expressive}, has garnered considerable attention due to its remarkable capability to harness the collective power of distributed quantum resources, surpassing the limitations of individual quantum nodes. 
Distributed quantum computation usually involves generating and transmitting quantum states across multiple nodes leveraging the advancements in quantum communication technologies~\cite{QKD_RevModPhys.92.025002}.
Remarkably, distributed quantum computing protocols offer a ray of hope in addressing privacy concerns in the presence of adversaries~\cite{QKD_RevModPhys.92.025002,Bennett2014,broadbent2009universal,Fitzsimons2017npjQI,polacchi2023multiclient}, while traditional classical methods have struggled to ensure the confidentiality of sensitive information during distributed processes.
These adversaries not only involve third-party attacks that can be tackled with well-celebrated quantum communication technologies such as quantum key distribution~\cite{QKD_RevModPhys.92.025002,Bennett2014}, but also include privacy concerns with untrusted computing nodes \cite{broadbent2009universal,Fitzsimons2017npjQI}.

A critical example of this vulnerability indeed lies in classical federated learning (FL)~\cite{mcmahan2017communication, yang2019federated}, where multiple clients collaboratively train a machine learning model to optimize a given task while keeping their training data distributed without being moved to a single server or data center. 
A central server is assigned the responsibility of aggregating the client model updates, typically the model cost function gradients generated by the clients using their local data. However, this opens up a possibility of leaking client's sensitive data to the server using gradient inversion attacks~\cite{zhao2020idlg, eloul2022enhancing, mothukuri2021survey, Zhu19, geiping2020inverting}.
While techniques employing homomorphic encryption or differential privacy~\cite{aono2017privacy,huang2021evaluating} have been introduced to tackle the problem, they usually demand additional computational and communication overhead or come at the expense of reduced model accuracy.
To this end, quantum technologies could provide a natural embedding of privacy. 
To counteract the gradient inversion attack, one recent proposal~\cite{kumar2023expressive} replaced the classical neural network in the FL model with variational quantum circuits built using expressive quantum feature maps such that the problem of a successful attack is reduced to solving high-degree multivariate Chebyshev equations. Other quantum-based proposals include adding a certain level of noise to the gradient values to reduce the probability of a successful gradient inversion attack~\cite{Noise_PRR2021}, leveraging blind quantum computing~\cite{Weikang2021}, and others \cite{ren2023QFLreview,chen2021, Huang2022,chu2023cryptoqfl}. 
An alternative to the aforementioned methods is to encode the client's classical gradient values into quantum states and leverage quantum communication between the clients and server to transmit the states. This provides opportunities to hide the gradient values of individual clients from the server while allowing the server to perform the model aggregation using appropriate quantum operations on their end. In this case, the transmitted quantum states offer an inherent advantage in terms of privacy even without additional privacy mechanisms, as the classical information can be encoded in logarithmic number of qubits and using Holevo's bound, the server could extract at most logarithmic number bits of classical information during each round of communication~\cite{gilboa2023exponential}. Moreover, we remark that the approach can be naturally integrated with quantum cryptographic techniques~\cite{QCryptoRevModPhys.74.145,QKD_RevModPhys.92.025002} to become robust against third-party attacks. 

In this work, we introduce protocols for the above approach, aiming to advance the capability of distributed quantum computing with quantum communication in the context of FL. Specifically, we propose two types of protocols: one based on \emph{private inner-product estimation} to perform model aggregation, and the other based on the concept of \emph{incremental learning} to encode the model aggregated sum in the phase of the quantum state. 
For the former, we transform the secure model aggregation task into a correlation estimation problem and generalize the recently-developed blind quantum bipartite correlator (BQBC) algorithm~\cite{Li2023blindQML} into multi-party scenarios. For $m$ clients with $d$ model parameters to be updated, the protocol involves a quantum communication cost of $\Tilde{\mathcal O}(md/\epsilon)$ where $\epsilon$ is the standard model update error, and it is quadratically better in $m$ compared to the analogous method based on classical secret sharing~\cite{bonawitz2016practical}. 
For the second type of protocol, similar to incremental learning, clients perform multi-party computation sequentially or simultaneously without having the server involved until the end of the protocol at which the server extracts the aggregated gradient information. For one of our proposed protocols within the framework of incremental learning, the secure multi-party summation algorithm achieves a similar quantum communication cost as the BQBC with the complexity being $\Tilde{\mathcal O}(m d/\epsilon)$.

These protocols are designed not only to bolster privacy but also to have an evaluation on the quantum communication costs. Through the application of quantum algorithms, this work aspires to unlock novel strategies that are capable of safeguarding sensitive information within the realm of distributed quantum computing while optimizing communication efficiency.
Furthermore, it is noteworthy that the suggested protocols can seamlessly integrate with quantum key distribution protocols, thereby ensuring information-theoretic security against external eavesdropper attacks.
Our work sheds light on designing efficient quantum communication-assisted federated learning algorithms and paves the way for secure distributed quantum machine learning protocols.


\begin{table*}[t!]
\centering
\caption{Privacy and communication complexity of proposed gradient-hidden quantum federated learning protocols. \label{table:alg_summary}}
\begin{tabular}{   m{4.6cm}| m{3.5cm} | m{4.2cm} |m{5.5cm} } 
  \hline \hline
  Protocol & Privacy mechanism  & Communication complexity~\footnote{\mbox{CC: Classical communication complexity. QC: Quantum communication complexity}} & Additional requirement \\ 
  \hline
   baseline (classical)     & classical secret sharing  & CC: $\mathcal O((m+m^2)d)$& classical communication among clients \\ 
  \hline
  inner product estimation with classical secret sharing  & classical secret sharing, amplitude encoding  & CC: $\mathcal O(m^2d)$ 
   
   QC: $\mathcal O(\frac{d\log m}{\epsilon^2})$ & classical communication among clients, quantum communication among clients\\ 
  \hline
 blind QBC algorithm & quantum encoding, $\quad$ random phase padding &  QC: $\mathcal O (\frac{md}{\epsilon}\log( m \log(\frac{m}{\epsilon})))$\footnote{\mbox{Additional classical  communication complexity $\mathcal O(m)$ when random phase padding is used.}} & quantum communication among clients\\ 
  \hline
   GHZ-based phase encoding & phase accumulation &  QC: $\mathcal O(\frac{md}{\epsilon^2} )$ & global entanglement\\ 
  \hline
    multiparty quantum summation & phase accumulation &   QC: $\mathcal O(\frac{md}{\epsilon} \log\frac{m}{\epsilon}) $ & quantum communication among clients\\
  \hline
  \hline
\end{tabular}
\end{table*}

\section{Problem Statement}

\subsection{Federated Learning setup}

We present the settings of the quantum communication-based federated learning scheme involving $m$ clients and a central server. Consider the setup with each client $i \in [m]$  having $N_i$ samples of the form,
$$\mathbf{X}^{(i)}, Y^{(i)} : \{(\mathbf{x}^{(i)}_j, y^{(i)}_j)\}_{j=1}^{N_i}, \hspace{2mm} i \in [m]$$
such that the total number of samples across all the clients is $N = \sum_{i\in [m]} N_i$. Here each $\mathbf{x}_j^{(i)} \in \mathbb{R}^n$ and $y_j^{(i)} \in \mathcal{C}$ for finite set of output classes.

The aim is to learn a single, global statistical model such that the client data is processed and stored locally, with only the intermediate model updates being communicated periodically with a central server. In particular, the goal is typically to minimize a central objective cost function,
\begin{equation}
   \text{min}_{\boldsymbol{\theta}} \left[\boldsymbol{\mathcal{L}}(\boldsymbol{\theta}) = \sum_{i=1}^{m} w_i\boldsymbol{\mathcal{L}}_i(\boldsymbol{\theta}) \right]
\end{equation}
where $\boldsymbol{\theta} = \{\theta_1, \cdots, \theta_d\} \in \mathbb{R}^d$ are the set of $d$ trainable parameters of the FL model. The user-defined term $w_i \geq 0$ determines the relative impact of each client in the global minimization procedure with the most natural setting being $w_i = \frac{N_i}{N}$. That is, here the weight $w_i$ depends on the local data size of individual clients and is known to both the server and clients. 

In the standard federated learning setup, at the $t$-th iteration, the clients each receive the parameter values $\boldsymbol{\theta}^t \in \mathbb{R}^d$ from the server and their task is to compute the gradients with respect to $\boldsymbol{\theta}^t$ and send it back to the server. Here the superscript denotes the iteration step. Upon performing a single batch training, they compute the $d$ gradient updates $\nabla\boldsymbol{\mathcal{L}}_i$ and share it with the server. The server's task is then to perform the gradient aggregation within a standard error bound $\epsilon$ to update the next set of parameters $\boldsymbol{\theta}^{t+1}$ using the rule,
\begin{equation}\label{eq:model_update_general}
   \boldsymbol{\theta}^{t+1} = \boldsymbol{\theta}^{t} - \alpha \sum_{i=1}^{m}w_i \nabla \boldsymbol{\mathcal{L}}_i(\boldsymbol{\theta^t}),
\end{equation}
where for the rest of the work, we assume the relative impact $w_i = \frac{N_i}{N}$, and $\alpha$ is the learning rate hyperparameter chosen by the server. The parameters $\boldsymbol{\theta}^{t+1}$ are then communicated back to the clients and the protocol repeats until a desired stopping criteria is reached.  

We denote that in many cases, one is interested in learning $\sum_{i=1}^{m}\nabla \boldsymbol{\mathcal{L}}_i(\boldsymbol{\theta^t}) \mod 2\pi$ as the model parameters can have a $2\pi$ period, particularly in quantum circuits. We point out that here the local circuit model of both the server and clients could be either classical or quantum, but they both have the capability of encoding their local data into quantum states. Further, we consider a quantum communication channel between the server and $m$ clients in order to facilitate the transmission of quantum states. 

\subsection{Data leakage in classical FL} \label{sec:data_leakage}

The existing classical FL setup was built on the premise that sharing gradients to the server would not leak the local data information to the server. However, this notion of privacy has been challenged by the wider community~\cite{mothukuri2021survey}. Specifically led by the work of~\cite{Zhu19} and followed up by~\cite{zhao2020idlg,GeipingBD020,Yin21,ZhuB21}, it's shown that it is possible for the \emph{honest-but-curious} server (who  strictly follows the protocol but is interested in learning clients' private data) to extract input data from model gradients. In fact, using the results of \cite{eloul2022enhancing}, we showcase in Appendix~\ref{app:gradient_inversion} how to easily invert the gradients generated from a fully connected neural network model to learn the data.

While classical techniques including homomorphic encryption~\cite{aono2017privacy} and secret sharing~\cite{bonawitz2016practical} have been employed to tackle the challenge, they usually impose a significant overhead in communication and computation cost, limiting their applications for federated learning tasks. On the other hand, randomization approach employing  differential privacy~\cite{huang2021evaluating}, while being simple to implement, usually leads to a reduced model accuracy and utility (see Appendix~\ref{app:classical_protection} for details).

In this work, we address the concern of data leakage originating from gradients that are 
generated by either a classical neural network based model or a variational quantum circuit based model~\cite{benedetti2019parameterized}.
The primary objective is to facilitate a secure global parameter update without divulging individual clients' gradient information $\nabla \boldsymbol{\mathcal{L}}_i(\boldsymbol{\theta^t})$ to the server, thereby mitigating the risk of gradient inversion attacks.
In order to hide the individual gradient information while still performing the model parameter update in Eq.~\ref{eq:model_update_general}, one can implement privacy in either multiplication between weights $w_i$ and local gradient $\nabla \boldsymbol{\mathcal{L}}_i(\boldsymbol{\theta^t})$
, or summation among weighted gradients. In what followings, we will show protocols along these two ways: secure inner product estimation or secure weighted gradient summation (in analogy with incremental learning).
Before diving into the details, we summarize the proposed protocols by listing the main privacy mechanism as well as quantum communication complexity and their requirements in Table~\ref{table:alg_summary}.

\section{Protocol I: secure inner product estimation}\label{sec:protocol_inner product}

In this section, we consider converting the model aggregation problem into task of distributed inner product estimation between server and clients where algorithms such as quantum bipartite correlator (QBC)~\cite{PhysRevLett.130.150602,Li2023blindQML} could be employed. 
From the federated learning parameter update rule Eq.~\ref{eq:model_update_general}, we note that for each parameter index, $j \in [d]$, the task for the server would be to perform multiplication between the weight $w_i$ and local gradient $\nabla \mathcal{L}_{i,j}(\boldsymbol{\theta})$, before summation of all weighted gradients to obtain $\theta_j^{t+1}$.

In the following, we start from a baseline approach where the secure inner product is performed with the assistance of classical secret sharing (CSS). Following it, we utilize the blind quantum bipartite correlator algorithm and propose a scheme for secure inner product estimation with quadratically fewer communication cost in $m$.

\subsection{Baseline: Classical secret sharing assisted inner-product estimation}\label{sec:baseline}

In this section, we start with a purely classical strategy to hide the gradients of the clients prior to sending the masked gradients to the server. We use this as a baseline to compare against the quantum gradient hiding strategies we develop over the next sections. The baseline strategy is built using the masking technique with one-time pads as introduced in Protocol 0 in Ref.~\cite{bonawitz2016practical}. For this protocol to succeed, we assume that each client is switched ``on" during the entirety of the protocol and further, has pairwise secure classical communication channels with each of the $m-1$ other clients. 

The protocol starts with each client $i$ sampling $m-1$ random values $s_{i,k} \in [0,R)$ for every other client indexed by $k$. Here $R$ is the chosen upper limit of the interval as agreed by all the clients. Similarly, all other clients generate the random values in $[0,R)$ for every other client. Next, clients $i$ and $k$ exchange $s_{i,k}$ and $s_{k,i}$ over their secure channel and compute the perturbations $p_{i,k} = s_{i,k} - s_{k,i} (\text{mod}~ R)$. We note that $p_{i,k} = -p_{k,i}$, Further, $p_{i,k} = 0$ when $i=k$. The clients repeat the above procedure a total of $d$ times (to mask each of the $d$ gradient values $\nabla \mathcal{L}_{i,j}(\boldsymbol{\theta})$). 

Next, for every parameter to be updated, each client sends masked gradient value to the server,
\begin{equation}
     y_{i} = \nabla \mathcal{L}_{i}(\boldsymbol{\theta}) + \frac{1}{w_i}\sum_{k=1}^m p_{i,k} (\text{mod}~ R).
    \label{eq:masked_grad}
\end{equation}
Note that we drop the parameter index $j$ hereafter for simplicity.
The task of the server is then to perform a weighted aggregation of the gradients in order to obtain the next set of parameter values. It can be trivially checked that an honest server always succeeds in performing the correct aggregation, i.e.,
\begin{equation}
\begin{split}
 \bar{y} &= \sum_{i=1}^m w_i y_{i} \\
    &= \sum_{i=1}^m w_i \nabla \mathcal{L}_{i}(\boldsymbol{\theta}) + \sum_{i,k}(s_{i,k} - s_{k,i}) (\text{mod}~ R)\\
    &= \sum_{i=1}^m w_i \nabla \mathcal{L}_{i}(\boldsymbol{\theta}).
\end{split}
\end{equation}
Further, privacy is guaranteed due to the use of one-time pad masking of gradients which guarantees information-theoretic security against malicious server.

The above scheme requires a total of $\frac{m(m-1)}{2}\times d \log(R)\approx \mathcal{O}(m^2 d)$ classical bits of communication between the clients and a further $\mathcal{O}(md)$ bits of communication between the clients and server to achieve secure aggregation.  Thus the total classical communication complexity required is,
\begin{equation}
    \mathcal C_{CSS}^C =  \mathcal O((m + m^2)d) \approx \mathcal{O}(m^2d).
\end{equation}

We remark that the above classical secret sharing based scheme could be augmented using quantum resources to provide a minor improvement in the total communication cost (Fig.~\ref{fig:BQBC_QFL}a). Specifically, after obtaining the masked gradients as in Eq.~\ref{eq:masked_grad}, the clients can collaboratively encode their masked gradients in an amplitude encoded quantum state,
\begin{equation}
    \ket{\phi_c} = \frac{1}{\mathcal N_c} \sum_{i=1}^m  y_{i}  \ket{i}, 
    \label{eq:grad_sup}
\end{equation}
where $\mathcal N_c$ is the normalization factor.
This state is then sent to the server which can recover the weighted aggregate sum by performing the SWAP test-based discrimination \cite{buhrman2001quantum} with their local state $ \ket{\phi_s} = \frac{1}{\mathcal N_s} \sum_{i=1}^m w_i \ket{i}$. Since the state in Eq.~\ref{eq:grad_sup} requires only $\mathcal{O}(\log(m))$ qubits, the amount of communication between the server and clients can be reduced to $\mathcal{O}(\log (m)/\epsilon^2)$, where $\epsilon$ is the error incurred in estimating the aggregated sum using the SWAP test. The total communication complexity with this scheme is,
\begin{equation}
    \mathcal C_{CSS}^Q =  \mathcal O\left(\left(\frac{\log m}{\epsilon^2} + m^2\right)d \right). 
\end{equation}

\begin{figure*}[t]
 \includegraphics[scale=0.53]{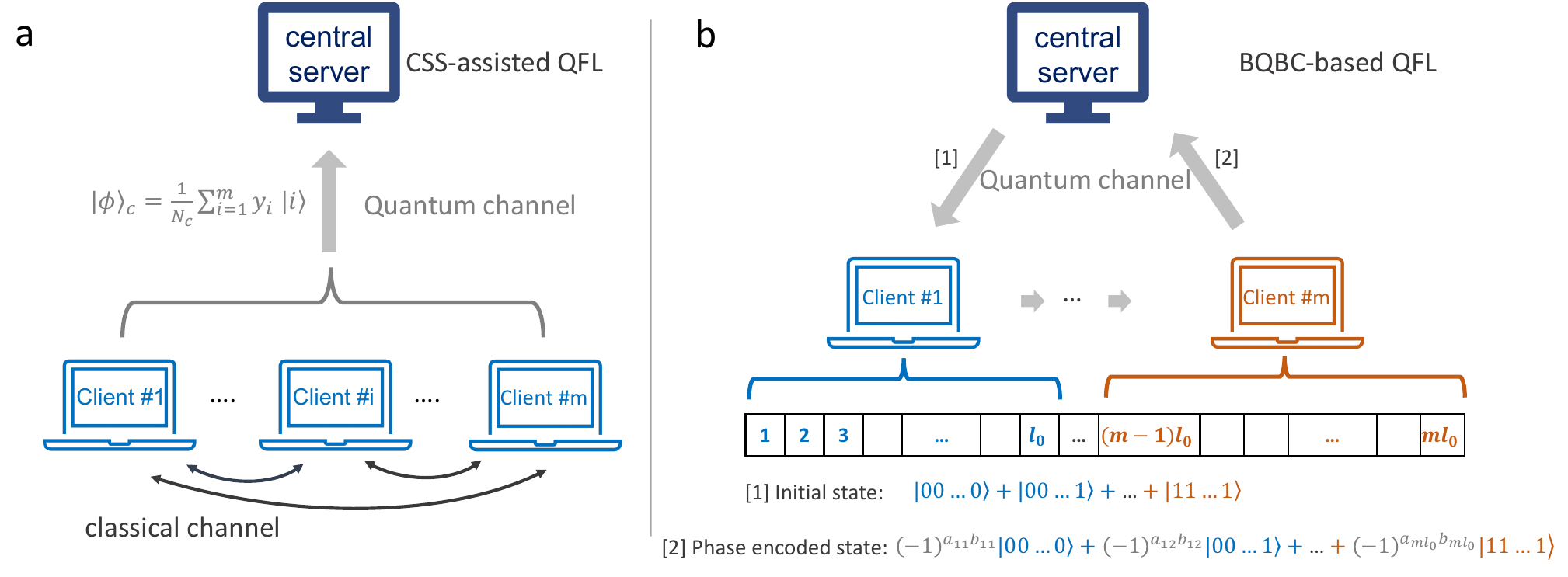}
\caption{Diagram of QFL protocols based on secure inner product estimation. \textbf{a.} CSS-assisted QFL protocol. The clients jointly prepare a state in which the amplitudes encode the masked gradients and then send it to the server. The gradient masking is achieved via classical secret sharing.
\textbf{b.} BQBC-based QFL protocol. We consider a central server with $m$ clients and there are quantum channels among them. During each round of communication, each client encodes their local gradient information in specific phases of the received state and then send back to the server. }
\label{fig:BQBC_QFL}
\end{figure*}

\subsection{Model aggregation with blind quantum bipartite correlator algorithm}\label{sec: BQBC_protocol}

To reduce the communication cost, in this section, we propose a method for model updating based on quantum bipartite correlator algorithm~\cite{PhysRevLett.130.150602,Li2023blindQML} that is designed to estimate inner product between remote vectors. 
The essential idea is a generalization of recently-proposed blind quantum bipartite correlator  algorithm~\cite{Li2023blindQML}: firstly, each client converts the gradient information into binary floating point numbers. Then, at each round of communication, the server passes the index qubit state that encodes weight information into each honest or honest-but-curious client and let them privately encode the gradient information into the phase of corresponding index qubits. Finally, the server receives back the index qubits and perform quantum counting algorithm to extract the desired aggregated gradient. 

We now proceed to the implementations details of the algorithm.
As mentioned, the goal is to have the sever performs the inner product estimation using the known weight information and gradient information that is only locally held by each client. 
For the $k$-th client ($k \in [m]$), both the weight $w_k$ and the gradient $\nabla\boldsymbol{\mathcal{L}}_k$ can be expanded as binary bitstrings both with size $l_k$: $\boldsymbol{a}_k$ and $\boldsymbol{b}_k$, such that $w_k \cdot \nabla\boldsymbol{\mathcal{L}}_k$ equals to the inner product $l_k \frac{1}{l_k}\sum_j a_{kj}b_{kj} = l_k \overline{a_k b_k}$.
One example of such expansions is to use the IEEE standard for floating-point arithmetic~\cite{IEEEstandard2019}, where we can have
\begin{equation}
    w_k = \sum_{i=0}^{l_k} 2^{u-i} a_{ki}, \nabla\boldsymbol{\mathcal{L}}_k = \sum_{i=0}^{l_k} 2^{v-i} b_{ki}.
\end{equation}
Here $u$ and $v$ are the highest digits of $w_k$ and  $\nabla\boldsymbol{\mathcal{L}}_k$, respectively, and are constants known to both server and clients.     We then get
\begin{equation}
    w_k \cdot \nabla\boldsymbol{\mathcal{L}}_k = \sum_{\lambda=0}^{2l_k} 2^{u+v-\lambda} \sum_{i=0}^{\lambda} a_{ki}b_{k(\lambda-i)}.
\end{equation}
In the following, we assume $l_k = l_0, \forall 1 \leq k \leq m$ for simplicity. Then the goal is to design a private inner product protocol to have the server evaluate $\sum_{k=1}^{m} \sum_{j}^{l_0} a_{kj} b_{kj}$. 

We thus consider the following protocol:
initially, the server prepares a quantum state with $\lceil \log (ml_0) \rceil$ index qubits $\frac{1}{\sqrt{2^{\lceil \log (ml_0) \rceil}}}\sum_{i=1}^{2^{\lceil \log (ml_0) \rceil}}\ket{i}$, and then applied controlled-gate to encode all $w_k$ information on a single qubit $o_a$. The final state is 
\begin{equation}
    \sum_{k=1}^m \sum_{i=1}^{l_0}\ket{k,i}\ket{a_{ki}}_{o_a},
\end{equation}
where the index $k$ denotes the $k$-th client and index $i$ is the index for bitstring with size $l_0$. We omit the normalization factor for above and following states in the protocol for simplicity.

Then, as shown in the diagram in Fig.~\ref{fig:BQBC_QFL}b, the server delivers the above $\lceil \log (ml_0) \rceil+1$ qubits to the first client. 
Note that a malicious server could prepare a state $\sum_{k=1}^m \sum_{i=1}^{l_0}c_{ki}\ket{k,i}\ket{a_{ki}}_{o_a}$  with non-uniform amplitude distribution of $c_{ki}$ to extract clients' information of interest. To detect such attacks, the first client would firstly decode the ancillary qubit $o_a$ (as the encoded weight information is known globally) and then measure the index qubits in $X$ basis. In the honest server case where $c_{ki}$ has a uniform distribution,  measurement outcome should be all +1.
 That is, if a malicious server tries to extract certain gradient information by increasing the amplitude of corresponding bitstrings, the index qubit state without the ancilla would not be $\ket{+}^{ \otimes \lceil \log (ml_0) \rceil }$.

After successful verification and re-encoding of the weight information in ancillary qubit $o_a$ , the first client encodes its local gradient information $\nabla\boldsymbol{\mathcal{L}}_1$ into the phase of the first part of index qubits, which leads to
\begin{equation}
    \sum_{i=1}^{l_0}(-1)^{a_{1i} b_{1i}}\ket{1,i}\ket{a_{1i}}_{o_a} +\sum_{k=2}^{m} \sum_{i=1}^{l_0} \ket{k,i}\ket{a_{ki}}_{o_a}.
\end{equation}
This could be done using CZ gate between qubit $o_a$ and a local qubit held by the first client that encodes $\nabla\boldsymbol{\mathcal{L}}_1$. 

The first client then passes the above state to the second client, who then encodes its local gradient information $\nabla\boldsymbol{\mathcal{L}}_2$ into the phase of the second part of index qubits. The resulting state is
\begin{equation}
    \begin{split}
        \sum_{i=1}^{l_0}(-1)^{a_{1i}b_{1i}}\ket{1,i}\ket{a_{1i}}_{o_a} + & \sum_{i=1}^{l_0}(-1)^{a_{2i} b_{2i}}\ket{2,i}\ket{a_{2i}}_{o_a} \\
        +&  \sum_{k=3}^{m} \sum_{i=1}^{l_0} \ket{k,i}\ket{a_{ki}}_{o_a} \\
    \end{split}
\end{equation}

The above process is repeated until all the clients have encoded their local gradient information in the phase, resulting to a state 
\begin{equation}
   \sum_{k=1}^{m} \sum_{i=1}^{l_0}  (-1)^{a_{ki} b_{ki}}\ket{k,i}\ket{a_{ki}}_{o_a}.
\end{equation}

Finally, the state is returned to the server by the last client. Then the server runs quantum counting algorithm~\cite{PhysRevLett.130.150602, Li2023blindQML} to evaluate $\frac{1}{m l_0}\sum_{k=1}^{m} \sum_{j=1}^{l_0} a_{kj} b_{kj} = \frac{1}{m}\sum_{k=1}^m  \overline{a_k b_k}$. 
In order to perform the estimation algorithm, $\mathcal O (\frac{1}{\epsilon})$ rounds of communication is needed where $\epsilon$ is the standard estimation error. We remark that quantum counting algorithm is based on Grover's search algorithm and is advantageous compared with SWAP-test based algorithms~\cite{PhysRevLett.124.060503} in terms of the error complexity.

We present the takeaway of this method here. Firstly, the privacy is encode in the index qubit states. When the server measures the index qubits, the probability of getting a specific index is simple $\frac{1}{m l_0}$, which is small when client number $m$ is large.  Moreover, the server could not amplify the amplitude of a specific index by preparing a uniformly distributed  superposition state, as the first client is capable of verifying it. Furthermore, even there are multi-round of communication and the server could perform collective attack, by increasing $l_0$ or adding a random pad on the phase, it's still hard for the server to get individual client's information~\cite{Li2023blindQML}. 
Note that here the privacy comes from the phase encoding, rather than summation of gradients as in incremental learning protocols.

The quantum communication complexity would be the total number of qubits transmitted in order to estimate $\frac{1}{m}\sum_{k=1}^m \overline{a_k b_k}$ in the protocol, which reads as
\begin{equation}
    \mathcal C_{BQBC} = \mathcal O (\frac{md\log(m l_0)}{\epsilon}) = \mathcal O (\frac{md}{\epsilon}\log( m \log(\frac{m}{\epsilon}))) .
\end{equation}
Again, here $m$ is the number of clients, and $l_0 $ is related with the precision of gradient $ml_0 = O(m\log(1/\epsilon_0)) = O( m \log(m/\epsilon))$ where $\epsilon_0$ is the inner product estimation error bound for single clients. This is better than classical secret sharing which has a total complexity in $O(m^2)$.
We note that in the absence of random phase padding, the technique here doesn't require classical communication at each round. The incorporation of random, one-time phase pads for privacy enhancement necessitates an additional classical communication cost of $\Tilde{\mathcal O}(m)$, as each client would need to send the padding information to the server at last.

\subsection{Redundant encoding}
The privacy of the protocol above could be further enhanced with redundant encoding of gradient data into binary bitstrings~\cite{Li2023blindQML}. In particular, we remark on the following theorem:

\begin{theorem}[Efficient redundant encoding]
\label{thm:redundant_encodinng}
In the BQBC-based QFL protocol,
given a fixed estimation error $\epsilon$, there exists a redundant encoding method with a redundant parameter $r$, such that the probability that server learns client's information decrease polynomially in $r$, which the communication complexity increases only poly-logarithmically in $r$.
\end{theorem}

\begin{proof}
Following Ref.~\cite{Li2023blindQML}, we consider the following redundant encoding approach aimed at reducing the probability that a malicious server acquiring a specific $b_i$ information with $i$ being the pertinent index of interest. 
 we describe the following protocol where both the client $k$ and server encode their single bit local information $b_{ki}$ and $a_{ki}$ into bitstrings ${\left[b^{\prime}_{ki,1},b^{\prime}_{ki,2},\cdots,b^{\prime}_{ki,r}\right]}$ and ${\left[a^{\prime}_{ki,1},a^{\prime}_{ki,2},\cdots,a^{\prime}_{ki,r}\right]}$ with size $r$, where $r$ is an integer and $r>1$. The total amount of bits then increases from $ml_0$ to $rml_0$. For the weight information, we consider the following encoding rule
\begin{equation}\label{eq:rule_extending_bitstring_size_X}
    \begin{split}
        \boldsymbol{a}^{\prime}_{ki,j} = a_{ki}; \quad k=1,2,...,m; i = 1,2,...,l_0; j = 1,2, ..., r;
    \end{split}
\end{equation}
which is a simply copy the bit $a_{ki}$ for $r$ times. Here $k$ is the index for client and $i$ is the index of bitstring held by each client.

On the other hand, for $\boldsymbol{b}^{\prime}_k$, the $k$-th client can hide the information $b_{ki}$ randomly in one of the $r$ digits and let the other $r-1$ digits to be all zero or one. That is, for bit index $i$, the $k$-th client chooses either
\begin{equation}\label{eq:rule_extending_bitstring_size_y1}
    \begin{split}
        &\boldsymbol{b}^{\prime}_{ki,j} = \delta_{j, R_{ki}} \cdot b_{ki}; \\
         i = 1,2,...,l_0; &j = 1,2, ..., r, R_{ki}\in\{1,2,...,r\}. \\
    \end{split}
\end{equation}
or 
\begin{equation}\label{eq:rule_extending_bitstring_size_y2}
    \begin{split}
         &\boldsymbol{b}^{\prime}_{ki,j} = (1-\delta_{j, R_{ki}}) + \delta_{j, R_{ki}} \cdot b_{ki}; \\
         i = 1,2,...,l_0; &j = 1,2, ..., r, R_{ki}\in\{1,2,...,r\}. \\
    \end{split}
\end{equation}
where $R_{ki}$ is an random number $R_{ki} \in [r]$ and $\delta_{j, R_{ki}}$ is the Kronecker delta function. 

Then, according to the above rules, by running the QBC algorithm, for each $k$ the server would get 
\begin{equation}
    \frac{1}{rl_0}\sum_i^{l_0} a_{ki} b_{ki} \quad \text{or} \quad  \frac{1}{rl_0}\sum_i^{l_0} a_{ki} b_{ki} + \frac{r-1}{rl_0}\sum_i^{l_0} a_{ki}
\end{equation}
depending on whether the $k$-th client chooses encoding method Eq.~\ref{eq:rule_extending_bitstring_size_y1} or Eq.~\ref{eq:rule_extending_bitstring_size_y2}. The difference between the two extracted values is $\frac{r-1}{rl_0}\sum_i^{l_0} a_{ki}$ that the server would know.
Note that this choice could vary for different client $k$. 
At the end of the protocol, each client can send an one-bit message via classical channel to the server and let server knows which one was used, after which the server could extract $\sum_{k=1}^{m} \sum_{j=1}^{l_0} a_{kj} b_{kj}$.  This process yields a classical communication $\mathcal O(m)$.

We remark that at each communication round, the probability that the server samples a specific bit reduces from $\frac{1}{ml_0}$ to $\frac{1}{rml_0}$ with $r>1$. Even though that $r$-times more communication round will be needed to achieve the same error bound $\epsilon$ as in the original QBC case, the server would not know which digit encodes the correct $b_{ki}$ information as here $R_{ki}$s are random numbers. Therefore, the probability that the server successfully gets a specific bit $b_{ki}$ would be 
\begin{equation}
    P = \frac{1}{ rm l_0} \times \frac{r}{\epsilon} \times \frac{1}{r} = \frac{1}{rm l_0  \epsilon},
\end{equation}
where the second term $\frac{r}{\epsilon}$ is the total number of communication rounds and the third term is $\frac{1}{r}$ is due to the randomness in $R_{ki}$. 

 It's clear to see that a larger value of $r$ corresponds to a decreased probability for the server to successfully extract valuable information from the client through the attack strategy. The flexibility that the client can independently choose encoding method also protects the majority information of $\boldsymbol{b}$, i.e., the client may choose Eq.~\ref{eq:rule_extending_bitstring_size_y1} to encode data if the majority of $\boldsymbol{b}$ is $1$ to decrease the probability that $1$s are being detected.
Nevertheless, the trade-off for employing this redundant encoding approach compared with the original one in Sec.~\ref{sec: BQBC_protocol} manifests as an augmented quantum communication complexity, as the transmitted qubit number goes from $\log (m l_0)$ to $\log (r m l_0)$ now.
To this end, with the above redundant encoding method, the communication complexity increases logarithmically in $r$, while the probability that server successfully gets a specific bit $b_{ki}$ decrease polynomially in $r$.     
\end{proof}

With Theorem~\ref{thm:redundant_encodinng}, we show that one can design a protocol such that the privacy goes polynomially better with a redundant encoding parameter while the communication cost only goes logarithmically or linearly with this parameter.

\section{Protocol 2: Incremental learning}

\begin{figure*}[t]
 \includegraphics[scale=0.56]{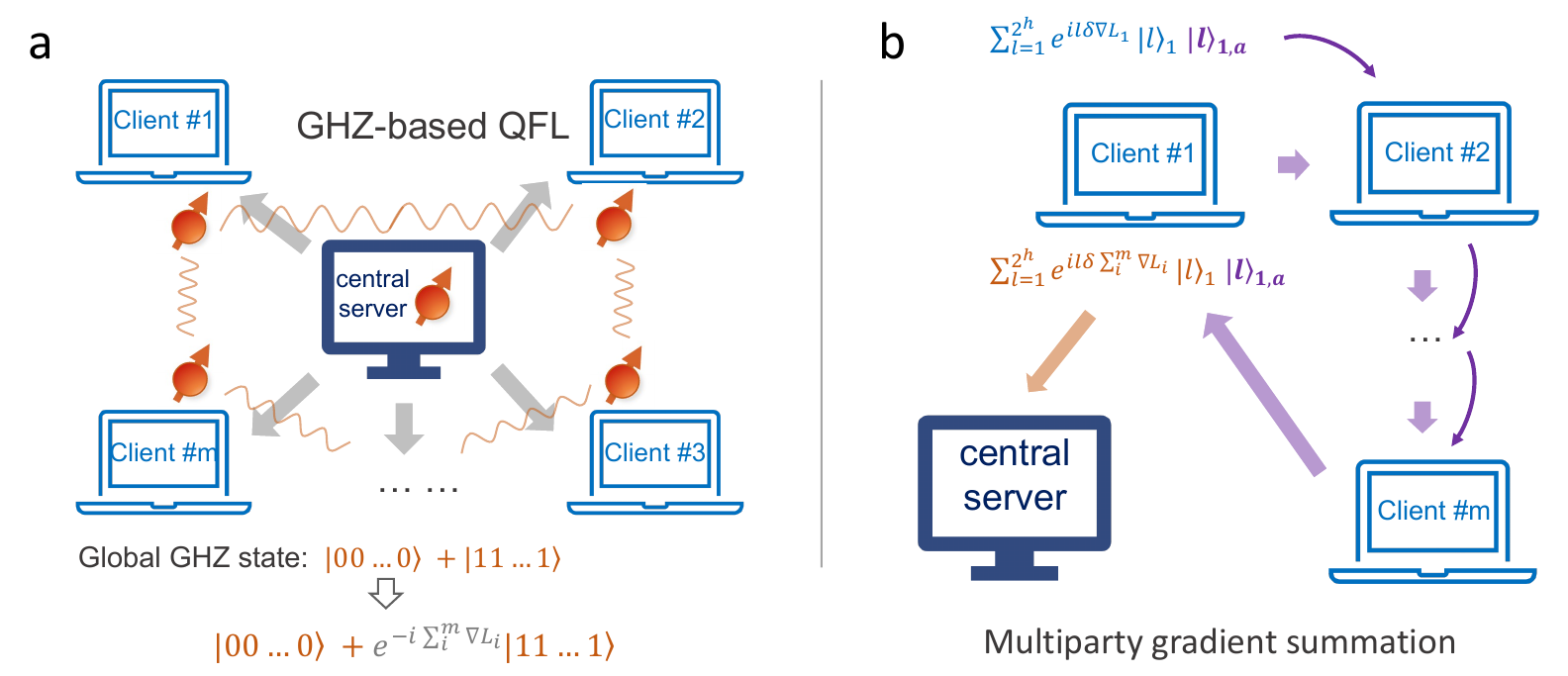}
\caption{Diagram of QFL protocols that are similar as incremental learning. \textbf{a.} Secure gradient aggregation based on global entanglement among clients. We consider GHZ states that are  distributed by the server or trusted client. After each client encodes its local gradient information, the server performs measurement on the phase of the state. \textbf{b.} Quantum federated learning with secure multiparty gradient summation. The ancillary $h$ qubits (in purple) are sent to the rest $(m-1)$ clients by the first client for gradient summation, after which the first client sends the other $h$-qubit state (in orange) to the server.  }
\label{fig:incremental_diagram}
\end{figure*}

The aforementioned protocols entail secure inner product estimation between the server and clients. 
An alternative approach involves the formulation of protocols with secure weighted gradient summation, ensuring that the server exclusively receives aggregated gradients rather than weighted gradients from individual clients. This concept aligns with the principles of incremental learning, wherein a model or parameter is iteratively trained or updated with new data, assimilating fresh information while preserving knowledge acquired from prior data. In this section, we introduce two such protocols: the first involves secure aggregation through a globally entangled state, while the second is grounded in secure multiparty gradient summation.

\subsection{Secure aggregation based on global entanglement}

We start by discussing a secure aggregation protocol using globally entangled state distributed among clients and server (Fig.~\ref{fig:incremental_diagram}a). Similar as Ref.~\cite{zhang2023federated}, at each time step and for each parameter to be updated, we consider a global ($m$+1)-qubit GHZ state and each qubit is held by one party. This can be achieved by letting the server prepare an GHZ state locally and then distribute the $m$-qubits to the $m$ clients via quantum channels. Alternatively, we can consider each party holds on local qubit and remote entanglement is generated via quantum photonic channels. 

Nevertheless, after the preparation of the GHZ state, for a given parameter, the $k$-th client encodes its local weighted gradient information $\nabla \mathcal{L}_k(\boldsymbol{\theta^t})$ into the phase of its local qubit by applying a phase gate. This process can be either sequentially or simultaneously.  The distributed qubits are then sent back to the server via the quantum channel. The resulting state now reads
\begin{equation}\label{eq:GHZ_global}
   \frac{1}{\sqrt{2}} ( \ket{00...0}_{s,1,2,...,m} + e^{-i \sum_k^m \nabla \mathcal{L}_k(\boldsymbol{\theta^t})}\ket{11...1}_{s,1,2,...,m} ).
\end{equation}
The server could firstly disentangle the $m$ received qubits by performing sequential CNOT gates between the local qubit $s$ with the rest $m$ qubits, leading to 
\begin{equation}
   \frac{1}{\sqrt{2}} ( \ket{0}_s + e^{-i \sum_k^m \nabla \mathcal{L}_k(\boldsymbol{\theta^t})}\ket{1}_s )\otimes \ket{0...0}_{1,2,...,m}.
\end{equation}
Then, similar as it's done in a typical Ramsey interferometry experiment~\cite{nielsen2010quantum}, the server would need to estimate the phase term $\sum_i^m \nabla \mathcal{L}_i(\boldsymbol{\theta^t})$ by applying a Hardmard gate on its local qubit followed by projective measurement in computational basis. The probability of getting zero would simply be
\begin{equation}
    P_k(0) = \frac{1+\cos (\sum_k^m \nabla \mathcal{L}_k(\boldsymbol{\theta^t}))}{2}. 
\end{equation}
The above process is repeated $\mathcal{O}(\frac{1}{\epsilon^2})$ times until the desired error bound $\epsilon$ is met. 
The procedure is iteratively applied to update all $d$ parameters.

We now perform security analysis of the gradient information. As all the local gradient is aggregated in the phase  of the GHZ state, the server could not extract the gradient of single clients. 
Under malicious server setting, in contrast of the semi-honest assumption in Ref.~\cite{zhang2023federated}, if the GHZ state is distributed by the server, the server could simply prepare the state where only the $j$-th client's qubit is entangled with the server qubit while others are not entangled, for example,
\begin{equation}
   \frac{1}{\sqrt{2}} (\ket{0}_1 + \ket{1}_1) \otimes ... \frac{1}{\sqrt{2}} (\ket{00}_{s,j} + \ket{11}_{s,j}) \otimes ... \frac{1}{\sqrt{2}} (\ket{0}_m + \ket{1}_m).
\end{equation}
Then, the malicious server would extract the gradient information of client $j$ by measuring the phase its local qubit. To tackle this adversary, the GHZ state could be distributed by a trusted client. 
Alternatively, the GHZ state can be generated by allowing communications among clients. That is, the honest (or honest-but-curious) clients could jointly prepares a $m$-qubit entangled state and then communicate with the server to reach the state Eq.~\ref{eq:GHZ_global}.

The total communication complexity for the aforementioned distributed entangled state scenario would be decided by the qubit distribution at each round and number of communication rounds to estimate the phase. Specifically, the total quantum communication cost would read
\begin{equation}
    \mathcal C_{GHZ} = \mathcal O(\frac{md}{\epsilon^2})
\end{equation}
with $\epsilon$ being the error bound for the phase estimation.  


\subsection{Secure multiparty gradient summation}

We next introduce a gradient-hidden quantum federated learning protocol using phase accumulation and estimation. Inspired by the secure multiparty quantum summation protocol proposed in Ref.~\cite{Shi2016}, we consider the following data encoding method. At each time step $t$, the gradient information  for parameter $k$ and client $l$ is $\nabla_k \mathcal{L}_l(\boldsymbol{\theta^t}) \in \{0, \delta, 2\delta, ... , 2\pi \}$. Note that here we set the upper bound of each individual gradient to be $2\pi$ for simplicity and the condition can be relaxed. 

As shown in the diagram in Fig.~\ref{fig:incremental_diagram}b, the protocol starts from the first client, who encodes its local gradient information for a given parameter into a $h$-qubit state $\ket{\nabla \mathcal{L}_1}$ with $h = \lceil \log (2\pi/\delta) \rceil$ (we consider $\log (2\pi/\delta)$ being an integer for simplity in the following). A subsequent quantum Fourier transform (QFT) would yield the following state 
\begin{equation}\label{eq:incremental_psi1}
    \ket{\psi}_{1} = \text{QFT} \ket{\nabla \mathcal{L}_1} = \frac{1}{\sqrt{2^h}} \sum_{l=1}^{2^h} e^{i l\delta \nabla \mathcal{L}_1  } \ket{l}_{1}.
\end{equation} 

Then, the first client prepares a $h$-qubit ancillary state that encodes the same information as $\ket{l}_{1}$ above. This can be achieved by simply applying CNOT gates between the first $h$ qubits in Eq.~\ref{eq:incremental_psi1} and the ancillary $h$ qubits. The resulting state reads
\begin{equation}\label{eq:incremental_psi2}
    \ket{\psi}_{1,a} = \frac{1}{\sqrt{2^h}} \sum_{l=1}^{2^h} e^{i l\delta \nabla \mathcal{L}_1  } \ket{l}_{1} \ket{l}_{1,a}
\end{equation}
where the subscript $a$ denotes ancilla. 
The first client then sends the $h$-qubit ancillary state to the second client via quantum communication. Similarly as the first client, the second client would first encode its local gradient information in another $h$-qubit state $\ket{\nabla \mathcal{L}_2}$. In order to perform the summation of $\nabla \mathcal{L}_1$ and $\nabla \mathcal{L}_2$, we consider the following operation for the second client: conditioned on the state of the received ancilla qubits, phase gates are applied on the local $h$ qubit, such that the resulting ensemble state reads
\begin{equation}\label{eq:incremental_psi3}
\begin{split}
    \ket{\psi}_{12,a} &= C_{a}Z_{2} \frac{1}{\sqrt{2^h}} \sum_{l=1}^{2^h} e^{i l\delta \nabla \mathcal{L}_1  } \ket{l}_{1} \ket{l}_{1,a} \ket{\nabla \mathcal{L}_2} \\
    & = \frac{1}{\sqrt{2^h}} \sum_{l=1}^{2^h} e^{i l\delta (\nabla \mathcal{L}_1 +\nabla \mathcal{L}_2)  } \ket{l}_{1} \ket{l}_{1,a} \ket{\nabla \mathcal{L}_2}. \\
\end{split}
\end{equation}
Note that the local state $\ket{\nabla_k \mathcal{L}_2}$ is not entangled with the rest of the system after the above operation. 
The second client would then pass the received $h$ qubits to the next client and the above process repeats until all the $m$ clients encode their local gradient information in the phase:
\begin{equation}\label{eq:incremental_psi4}
\begin{split}
     \ket{\psi}_{123...m,a} =  \frac{1}{\sqrt{2^h}} \sum_{l=1}^{2^h} e^{i l\delta  \sum_i^m \nabla \mathcal{L}_i  } & \ket{l}_{1} \ket{l}_{1,a} \ket{\nabla \mathcal{L}_2} \\
     & ...\ket{\nabla \mathcal{L}_{m-1}} \ket{\nabla \mathcal{L}_m}. \\
\end{split}
\end{equation}
The $m$-th cient would return the ancillary qubits back to the first client, who will subsequently perform verification on the states to detect potential dishonesty of the involved parities. Specifically, the first client would first uncompute the ancillary qubits with CNOT gates, leading to 
\begin{equation}\label{eq:incremental_psi5}
\begin{split}
     \ket{\psi}_{123...m,a^{'}} =  \frac{1}{\sqrt{2^h}} \sum_{l=1}^{2^h} e^{i l\delta  \sum_i^m \nabla \mathcal{L}_i  } &\ket{l}_{1} \ket{0}_{1,a} \ket{\nabla \mathcal{L}_2} \\
     & ...\ket{\nabla \mathcal{L}_{m-1}} \ket{\nabla \mathcal{L}_m}. \\
\end{split}
\end{equation}
Then the ancillary qubits are measured. In the absence of malicious client that tries to extract the phase information of previous clients and perform projective measurements on the ancillary qubits in computational basis, the measurement should yield 0 for all the $h$ qubits. For example, if a malicious client applies inverse QFT on the ancillary qubits to extract the aggregated phase, the first client can detect this anomaly, given that the ancillary qubits cannot be reset to $\ket{0}_{1,a}$ in such a scenario.


Upon the verification, the first client would send the $\ket{l}_1$ state to the server, who can then apply inverse QFT to extract the accumulated gradient information.
\begin{equation}\label{eq:incremental_psi6}
\begin{split}
    \ket{\psi}_s &= \text{QFT}^{-1} \frac{1}{\sqrt{2^h}} \sum_{l=1}^{2^h} e^{i l\delta  \sum_i^m \nabla \mathcal{L}_i  } \ket{l}_{1} \\
    & = \ket{   \sum_i^m \nabla \mathcal{L}_i  \mod 2\pi}_1   \\
\end{split}
\end{equation}
With the state $\ket{   \sum_i^m \nabla_k \mathcal{L}_i  \mod 2\pi}_1$ outlined above, the server could perform model aggregation and update the model accordingly. The same protocol applies for other parameters to be updated and different time windows. 

We remark that as the state received by the server is $\frac{1}{\sqrt{2^h}} \sum_{l=1}^{2^h} e^{i l\delta  \sum_i^m \nabla_k \mathcal{L}_i  } \ket{l}_{1}$ and the gradient aggregation has already been performed incrementally, the server could not extract the local gradient information held 
by individual clients. Moreover, the verification procedure could ensure that a malicious client could not simply perform measurement on the phase of the received qubits to extract the previously aggregated gradient information. To this end, the protocol relies on one trusted client node that can prepare the entangled state in Eq.~\ref{eq:incremental_psi2} and send the final $h$-qubit state to the server. 
The efficiency of this incremental learning protocol might be improved by pre-assigning clients into multiple batches in which there are at least one trusted node. 

We now discuss the quantum communication cost of the designed protocol. As discussed above, the $h$-qubit states are transmitted among all the $m$ clients for each parameter to be updated. For $p$ iterations of the process that are needed to yield a standard error $\epsilon$ on the phase, the communication complexity  of this secure multiparty summation protocol is given by
\begin{equation}
    \mathcal C_{SMS} = (m+1)\times d \times h \times p =\mathcal O (\frac{md}{\epsilon} \log\frac{m}{\epsilon}) 
\end{equation}
where the $\log(\frac{m}{\epsilon})$ term is associated with the error that comes from assigning $\nabla_k \mathcal{L}(\boldsymbol{\theta^t}) \in \{0, \delta, 2\delta, ... , 2\pi \}$.

Alternatively, one may consider encoding all the $d$ gradient information in a superposition state with $\lceil \log d \rceil$ qubits to reduce communication cost in terms of number of parameters $d$. However, as the server would need to update the $d$ parameters separately, at least $d$ samplings are required to query the encoded gradient information hence the total communication complexity in $d$ would still be $\Tilde{\mathcal{O}}(d)$.

\section{Discussions}

We remark that the above protocols leveraging quantum communication can be integrated with common quantum cryptography techniques~\cite{QCryptoRevModPhys.74.145,DecoyPhysRevLett.94.230504,QKD_RevModPhys.92.025002,Sheng2017SB} to be secure against external attacks. As an example, we consider using decoy state~\cite{DecoyPhysRevLett.94.230504} 
to detect eavesdropping attacks: when a quantum state with $n$ data qubits is sent to another party via a quantum channel during the QFL protocols, 
decoy states are randomly inserted and sent along with data qubits.

More specifically, when a $n$-qubit state is transmitted, we consider $n_d = \mathcal O(n)$  decoy qubits that are randomly drawn from $\{ \ket{0},\ket{1},\ket{+},\ket{-}  \}$ by the sender. The receiver receives the data and decoy qubits from the quantum channel, as well as positions and encoding basis of decoy qubits from a separated classical channel. After measuring the decoy qubits in the instructed basis, the receiver transmits the measurement results to the sender, who will then calculate the error rate and detect the potential existence of external eavesdropper. In this simple case, for a given decoy state, the probability that the eavesdropper performs a measurement on it without being detected is simply $\frac{3}{4}$ and the probability drops exponentially when there are $n_d$ uncorrelated decoy qubits. Advanced decoy-state quantum key distribution techniques~\cite{QKD_RevModPhys.92.025002} can be implemented to enhance the protocol's resilience against third-party attacks.

In this scenario, while the protocol demonstrates the capability to detect eavesdropper attacks, it incurs an additional cost in communication complexity. Specifically, there is an extra classical and quantum communication cost of $\mathcal{O}(n_d)$ between each sender and receiver pair.


It is important to emphasize that the proposed QFL protocols do not depend on a variational quantum circuit for gradient generation. Instead, gradient information can be produced using a classical neural network, thereby reducing the quantum capability demands on both the server and clients. Furthermore, in contrast to numerous classical federated learning algorithms that may face a trade-off between privacy loss and utility loss~\cite{zhang2022free}, the quantum protocols presented in this study do not compromise privacy for diminished utility, such as reduced accuracy.

\section{Conclusion}

In conclusion, we design gradient-hidden protocols for secure federated learning to protect against gradient inversion attacks and safeguard clients' local information. The proposed algorithms involve quantum communication among a server and clients, and we analyze both privacy and communication costs. The secure inner product estimation protocol based on BQBC relies on transmitting a logarithmic number of qubits to reduce the information server could query. We devise an efficient redundant encoding method to improve privacy further. For the incremental learning protocols, we consider both phase encoding based on globally entangled state and secure multi-party summation of gradient information to prevent the server from learning individual gradients from clients. We further discuss the quantum and classical communication costs involved in each protocol.

Our present study suggests numerous potential avenues for future research. Firstly, while the proposed protocols primarily address adversaries in the form of a malicious or honest-but-curious server, there is a need to develop secure protocols tailored to scenarios involving a dishonest majority, encompassing malicious clients~\cite{Dulek2020Dihonest_majority,Xia2021ByzantineAttack}. Secondly, the protocols proposed herein can be extended and applied to other secure distributed quantum computing tasks, such as quantum e-voting protocols~\cite{Arapinis_2021E-voting,Centrone_2022E-voting}. Furthermore, our work would motivate subsequent efforts aimed at achieving quantum communication advantages while preserving privacy advantages over classical counterparts in practical distributed machine learning tasks~\cite{PhysRevLett.117.100502,gilboa2023exponential}.
To this end, our work sheds light on designing efficient quantum communication-assisted distributed machine learning algorithms, studying quantum inherent privacy mechanisms, and paves the way for secure distributed quantum
computing protocols.

\acknowledgements
The authors thank Shaltiel Eloul, Jamie Heredge and other colleagues at the
Global Technology Applied Research Center of JPMorgan Chase for support and helpful discussions.

\bibliography{draft_arXiv} 

\begin{thebibliography}{57}%
\makeatletter
\providecommand \@ifxundefined [1]{%
 \@ifx{#1\undefined}
}%
\providecommand \@ifnum [1]{%
 \ifnum #1\expandafter \@firstoftwo
 \else \expandafter \@secondoftwo
 \fi
}%
\providecommand \@ifx [1]{%
 \ifx #1\expandafter \@firstoftwo
 \else \expandafter \@secondoftwo
 \fi
}%
\providecommand \natexlab [1]{#1}%
\providecommand \enquote  [1]{``#1''}%
\providecommand \bibnamefont  [1]{#1}%
\providecommand \bibfnamefont [1]{#1}%
\providecommand \citenamefont [1]{#1}%
\providecommand \href@noop [0]{\@secondoftwo}%
\providecommand \href [0]{\begingroup \@sanitize@url \@href}%
\providecommand \@href[1]{\@@startlink{#1}\@@href}%
\providecommand \@@href[1]{\endgroup#1\@@endlink}%
\providecommand \@sanitize@url [0]{\catcode `\\12\catcode `\$12\catcode
  `\&12\catcode `\#12\catcode `\^12\catcode `\_12\catcode `\%12\relax}%
\providecommand \@@startlink[1]{}%
\providecommand \@@endlink[0]{}%
\providecommand \url  [0]{\begingroup\@sanitize@url \@url }%
\providecommand \@url [1]{\endgroup\@href {#1}{\urlprefix }}%
\providecommand \urlprefix  [0]{URL }%
\providecommand \Eprint [0]{\href }%
\providecommand \doibase [0]{https://doi.org/}%
\providecommand \selectlanguage [0]{\@gobble}%
\providecommand \bibinfo  [0]{\@secondoftwo}%
\providecommand \bibfield  [0]{\@secondoftwo}%
\providecommand \translation [1]{[#1]}%
\providecommand \BibitemOpen [0]{}%
\providecommand \bibitemStop [0]{}%
\providecommand \bibitemNoStop [0]{.\EOS\space}%
\providecommand \EOS [0]{\spacefactor3000\relax}%
\providecommand \BibitemShut  [1]{\csname bibitem#1\endcsname}%
\let\auto@bib@innerbib\@empty
\bibitem [{\citenamefont {Cuomo}\ \emph {et~al.}(2020)\citenamefont {Cuomo},
  \citenamefont {Caleffi},\ and\ \citenamefont {Cacciapuoti}}]{Cuomo2020DQC}%
  \BibitemOpen
  \bibfield  {author} {\bibinfo {author} {\bibfnamefont {D.}~\bibnamefont
  {Cuomo}}, \bibinfo {author} {\bibfnamefont {M.}~\bibnamefont {Caleffi}},\
  and\ \bibinfo {author} {\bibfnamefont {A.~S.}\ \bibnamefont {Cacciapuoti}},\
  }\bibfield  {title} {\bibinfo {title} {Towards a distributed quantum
  computing ecosystem},\ }\href {https://doi.org/10.1049/iet-qtc.2020.0002}
  {\bibfield  {journal} {\bibinfo  {journal} {{IET} Quantum Communication}\
  }\textbf {\bibinfo {volume} {1}},\ \bibinfo {pages} {3} (\bibinfo {year}
  {2020})}\BibitemShut {NoStop}%
\bibitem [{\citenamefont {Caleffi}\ \emph {et~al.}(2022)\citenamefont
  {Caleffi}, \citenamefont {Amoretti}, \citenamefont {Ferrari}, \citenamefont
  {Cuomo}, \citenamefont {Illiano}, \citenamefont {Manzalini},\ and\
  \citenamefont {Cacciapuoti}}]{Caleffi2022distributed}%
  \BibitemOpen
  \bibfield  {author} {\bibinfo {author} {\bibfnamefont {M.}~\bibnamefont
  {Caleffi}}, \bibinfo {author} {\bibfnamefont {M.}~\bibnamefont {Amoretti}},
  \bibinfo {author} {\bibfnamefont {D.}~\bibnamefont {Ferrari}}, \bibinfo
  {author} {\bibfnamefont {D.}~\bibnamefont {Cuomo}}, \bibinfo {author}
  {\bibfnamefont {J.}~\bibnamefont {Illiano}}, \bibinfo {author} {\bibfnamefont
  {A.}~\bibnamefont {Manzalini}},\ and\ \bibinfo {author} {\bibfnamefont
  {A.~S.}\ \bibnamefont {Cacciapuoti}},\ }\href@noop {} {\bibinfo {title}
  {Distributed quantum computing: a survey}} (\bibinfo {year} {2022}),\ \Eprint
  {https://arxiv.org/abs/2212.10609} {arXiv:2212.10609 [quant-ph]} \BibitemShut
  {NoStop}%
\bibitem [{\citenamefont {Beals}\ \emph {et~al.}(2013)\citenamefont {Beals},
  \citenamefont {Brierley}, \citenamefont {Gray}, \citenamefont {Harrow},
  \citenamefont {Kutin}, \citenamefont {Linden}, \citenamefont {Shepherd},\
  and\ \citenamefont {Stather}}]{Beals2013}%
  \BibitemOpen
  \bibfield  {author} {\bibinfo {author} {\bibfnamefont {R.}~\bibnamefont
  {Beals}}, \bibinfo {author} {\bibfnamefont {S.}~\bibnamefont {Brierley}},
  \bibinfo {author} {\bibfnamefont {O.}~\bibnamefont {Gray}}, \bibinfo {author}
  {\bibfnamefont {A.~W.}\ \bibnamefont {Harrow}}, \bibinfo {author}
  {\bibfnamefont {S.}~\bibnamefont {Kutin}}, \bibinfo {author} {\bibfnamefont
  {N.}~\bibnamefont {Linden}}, \bibinfo {author} {\bibfnamefont
  {D.}~\bibnamefont {Shepherd}},\ and\ \bibinfo {author} {\bibfnamefont
  {M.}~\bibnamefont {Stather}},\ }\bibfield  {title} {\bibinfo {title}
  {Efficient distributed quantum computing},\ }\href
  {https://doi.org/10.1098/rspa.2012.0686} {\bibfield  {journal} {\bibinfo
  {journal} {Proceedings of the Royal Society A: Mathematical, Physical and
  Engineering Sciences}\ }\textbf {\bibinfo {volume} {469}},\ \bibinfo {pages}
  {20120686} (\bibinfo {year} {2013})}\BibitemShut {NoStop}%
\bibitem [{\citenamefont {Cacciapuoti}\ \emph {et~al.}(2020)\citenamefont
  {Cacciapuoti}, \citenamefont {Caleffi}, \citenamefont {Tafuri}, \citenamefont
  {Cataliotti}, \citenamefont {Gherardini},\ and\ \citenamefont
  {Bianchi}}]{Cacciapuoti2020}%
  \BibitemOpen
  \bibfield  {author} {\bibinfo {author} {\bibfnamefont {A.~S.}\ \bibnamefont
  {Cacciapuoti}}, \bibinfo {author} {\bibfnamefont {M.}~\bibnamefont
  {Caleffi}}, \bibinfo {author} {\bibfnamefont {F.}~\bibnamefont {Tafuri}},
  \bibinfo {author} {\bibfnamefont {F.~S.}\ \bibnamefont {Cataliotti}},
  \bibinfo {author} {\bibfnamefont {S.}~\bibnamefont {Gherardini}},\ and\
  \bibinfo {author} {\bibfnamefont {G.}~\bibnamefont {Bianchi}},\ }\bibfield
  {title} {\bibinfo {title} {Quantum internet: Networking challenges in
  distributed quantum computing},\ }\href
  {https://doi.org/10.1109/mnet.001.1900092} {\bibfield  {journal} {\bibinfo
  {journal} {{IEEE} Network}\ }\textbf {\bibinfo {volume} {34}},\ \bibinfo
  {pages} {137} (\bibinfo {year} {2020})}\BibitemShut {NoStop}%
\bibitem [{\citenamefont {Montanaro}\ and\ \citenamefont
  {Shao}(2023)}]{montanaro2023quantum}%
  \BibitemOpen
  \bibfield  {author} {\bibinfo {author} {\bibfnamefont {A.}~\bibnamefont
  {Montanaro}}\ and\ \bibinfo {author} {\bibfnamefont {C.}~\bibnamefont
  {Shao}},\ }\href@noop {} {\bibinfo {title} {Quantum communication complexity
  of linear regression}} (\bibinfo {year} {2023}),\ \Eprint
  {https://arxiv.org/abs/2210.01601} {arXiv:2210.01601 [quant-ph]} \BibitemShut
  {NoStop}%
\bibitem [{\citenamefont {Gilboa}\ and\ \citenamefont
  {McClean}(2023)}]{gilboa2023exponential}%
  \BibitemOpen
  \bibfield  {author} {\bibinfo {author} {\bibfnamefont {D.}~\bibnamefont
  {Gilboa}}\ and\ \bibinfo {author} {\bibfnamefont {J.~R.}\ \bibnamefont
  {McClean}},\ }\href@noop {} {\bibinfo {title} {Exponential quantum
  communication advantage in distributed learning}} (\bibinfo {year} {2023}),\
  \Eprint {https://arxiv.org/abs/2310.07136} {arXiv:2310.07136 [quant-ph]}
  \BibitemShut {NoStop}%
\bibitem [{\citenamefont {Li}\ \emph {et~al.}(2023)\citenamefont {Li},
  \citenamefont {Li}, \citenamefont {Amer}, \citenamefont {Shaydulin},
  \citenamefont {Chakrabarti}, \citenamefont {Wang}, \citenamefont {Xu},
  \citenamefont {Tang}, \citenamefont {Schoch}, \citenamefont {Kumar},
  \citenamefont {Lim}, \citenamefont {Li}, \citenamefont {Cappellaro},\ and\
  \citenamefont {Pistoia}}]{Li2023blindQML}%
  \BibitemOpen
  \bibfield  {author} {\bibinfo {author} {\bibfnamefont {C.}~\bibnamefont
  {Li}}, \bibinfo {author} {\bibfnamefont {B.}~\bibnamefont {Li}}, \bibinfo
  {author} {\bibfnamefont {O.}~\bibnamefont {Amer}}, \bibinfo {author}
  {\bibfnamefont {R.}~\bibnamefont {Shaydulin}}, \bibinfo {author}
  {\bibfnamefont {S.}~\bibnamefont {Chakrabarti}}, \bibinfo {author}
  {\bibfnamefont {G.}~\bibnamefont {Wang}}, \bibinfo {author} {\bibfnamefont
  {H.}~\bibnamefont {Xu}}, \bibinfo {author} {\bibfnamefont {H.}~\bibnamefont
  {Tang}}, \bibinfo {author} {\bibfnamefont {I.}~\bibnamefont {Schoch}},
  \bibinfo {author} {\bibfnamefont {N.}~\bibnamefont {Kumar}}, \bibinfo
  {author} {\bibfnamefont {C.}~\bibnamefont {Lim}}, \bibinfo {author}
  {\bibfnamefont {J.}~\bibnamefont {Li}}, \bibinfo {author} {\bibfnamefont
  {P.}~\bibnamefont {Cappellaro}},\ and\ \bibinfo {author} {\bibfnamefont
  {M.}~\bibnamefont {Pistoia}},\ }\href@noop {} {\bibinfo {title} {Blind
  quantum machine learning with quantum bipartite correlator}} (\bibinfo {year}
  {2023}),\ \Eprint {https://arxiv.org/abs/2310.12893} {arXiv:2310.12893
  [quant-ph]} \BibitemShut {NoStop}%
\bibitem [{\citenamefont {Tang}\ \emph {et~al.}(2023)\citenamefont {Tang},
  \citenamefont {Li}, \citenamefont {Wang}, \citenamefont {Xu}, \citenamefont
  {Li}, \citenamefont {Barr}, \citenamefont {Cappellaro},\ and\ \citenamefont
  {Li}}]{PhysRevLett.130.150602}%
  \BibitemOpen
  \bibfield  {author} {\bibinfo {author} {\bibfnamefont {H.}~\bibnamefont
  {Tang}}, \bibinfo {author} {\bibfnamefont {B.}~\bibnamefont {Li}}, \bibinfo
  {author} {\bibfnamefont {G.}~\bibnamefont {Wang}}, \bibinfo {author}
  {\bibfnamefont {H.}~\bibnamefont {Xu}}, \bibinfo {author} {\bibfnamefont
  {C.}~\bibnamefont {Li}}, \bibinfo {author} {\bibfnamefont {A.}~\bibnamefont
  {Barr}}, \bibinfo {author} {\bibfnamefont {P.}~\bibnamefont {Cappellaro}},\
  and\ \bibinfo {author} {\bibfnamefont {J.}~\bibnamefont {Li}},\ }\bibfield
  {title} {\bibinfo {title} {Communication-efficient quantum algorithm for
  distributed machine learning},\ }\href
  {https://doi.org/10.1103/PhysRevLett.130.150602} {\bibfield  {journal}
  {\bibinfo  {journal} {Phys. Rev. Lett.}\ }\textbf {\bibinfo {volume} {130}},\
  \bibinfo {pages} {150602} (\bibinfo {year} {2023})}\BibitemShut {NoStop}%
\bibitem [{\citenamefont {Kumar}\ \emph {et~al.}(2023)\citenamefont {Kumar},
  \citenamefont {Heredge}, \citenamefont {Li}, \citenamefont {Eloul},
  \citenamefont {Sureshbabu},\ and\ \citenamefont
  {Pistoia}}]{kumar2023expressive}%
  \BibitemOpen
  \bibfield  {author} {\bibinfo {author} {\bibfnamefont {N.}~\bibnamefont
  {Kumar}}, \bibinfo {author} {\bibfnamefont {J.}~\bibnamefont {Heredge}},
  \bibinfo {author} {\bibfnamefont {C.}~\bibnamefont {Li}}, \bibinfo {author}
  {\bibfnamefont {S.}~\bibnamefont {Eloul}}, \bibinfo {author} {\bibfnamefont
  {S.~H.}\ \bibnamefont {Sureshbabu}},\ and\ \bibinfo {author} {\bibfnamefont
  {M.}~\bibnamefont {Pistoia}},\ }\href@noop {} {\bibinfo {title} {Expressive
  variational quantum circuits provide inherent privacy in federated learning}}
  (\bibinfo {year} {2023}),\ \Eprint {https://arxiv.org/abs/2309.13002}
  {arXiv:2309.13002 [quant-ph]} \BibitemShut {NoStop}%
\bibitem [{\citenamefont {Xu}\ \emph {et~al.}(2020)\citenamefont {Xu},
  \citenamefont {Ma}, \citenamefont {Zhang}, \citenamefont {Lo},\ and\
  \citenamefont {Pan}}]{QKD_RevModPhys.92.025002}%
  \BibitemOpen
  \bibfield  {author} {\bibinfo {author} {\bibfnamefont {F.}~\bibnamefont
  {Xu}}, \bibinfo {author} {\bibfnamefont {X.}~\bibnamefont {Ma}}, \bibinfo
  {author} {\bibfnamefont {Q.}~\bibnamefont {Zhang}}, \bibinfo {author}
  {\bibfnamefont {H.-K.}\ \bibnamefont {Lo}},\ and\ \bibinfo {author}
  {\bibfnamefont {J.-W.}\ \bibnamefont {Pan}},\ }\bibfield  {title} {\bibinfo
  {title} {Secure quantum key distribution with realistic devices},\ }\href
  {https://doi.org/10.1103/RevModPhys.92.025002} {\bibfield  {journal}
  {\bibinfo  {journal} {Rev. Mod. Phys.}\ }\textbf {\bibinfo {volume} {92}},\
  \bibinfo {pages} {025002} (\bibinfo {year} {2020})}\BibitemShut {NoStop}%
\bibitem [{\citenamefont {Bennett}\ and\ \citenamefont
  {Brassard}(2014)}]{Bennett2014}%
  \BibitemOpen
  \bibfield  {author} {\bibinfo {author} {\bibfnamefont {C.~H.}\ \bibnamefont
  {Bennett}}\ and\ \bibinfo {author} {\bibfnamefont {G.}~\bibnamefont
  {Brassard}},\ }\bibfield  {title} {\bibinfo {title} {Quantum cryptography:
  Public key distribution and coin tossing},\ }\href
  {https://doi.org/10.1016/j.tcs.2014.05.025} {\bibfield  {journal} {\bibinfo
  {journal} {Theoretical Computer Science}\ }\textbf {\bibinfo {volume}
  {560}},\ \bibinfo {pages} {7} (\bibinfo {year} {2014})}\BibitemShut {NoStop}%
\bibitem [{\citenamefont {Broadbent}\ \emph {et~al.}(2009)\citenamefont
  {Broadbent}, \citenamefont {Fitzsimons},\ and\ \citenamefont
  {Kashefi}}]{broadbent2009universal}%
  \BibitemOpen
  \bibfield  {author} {\bibinfo {author} {\bibfnamefont {A.}~\bibnamefont
  {Broadbent}}, \bibinfo {author} {\bibfnamefont {J.}~\bibnamefont
  {Fitzsimons}},\ and\ \bibinfo {author} {\bibfnamefont {E.}~\bibnamefont
  {Kashefi}},\ }\bibfield  {title} {\bibinfo {title} {Universal blind quantum
  computation},\ }in\ \href {https://doi.org/10.1109/focs.2009.36} {\emph
  {\bibinfo {booktitle} {2009 50th Annual IEEE Symposium on Foundations of
  Computer Science}}}\ (\bibinfo  {publisher} {IEEE},\ \bibinfo {year}
  {2009})\BibitemShut {NoStop}%
\bibitem [{\citenamefont {Fitzsimons}(2017)}]{Fitzsimons2017npjQI}%
  \BibitemOpen
  \bibfield  {author} {\bibinfo {author} {\bibfnamefont {J.~F.}\ \bibnamefont
  {Fitzsimons}},\ }\bibfield  {title} {\bibinfo {title} {Private quantum
  computation: an introduction to blind quantum computing and related
  protocols},\ }\bibfield  {journal} {\bibinfo  {journal} {npj Quantum
  Information}\ }\textbf {\bibinfo {volume} {3}},\ \href
  {https://doi.org/10.1038/s41534-017-0025-3} {10.1038/s41534-017-0025-3}
  (\bibinfo {year} {2017})\BibitemShut {NoStop}%
\bibitem [{\citenamefont {Polacchi}\ \emph {et~al.}(2023)\citenamefont
  {Polacchi}, \citenamefont {Leichtle}, \citenamefont {Limongi}, \citenamefont
  {Carvacho}, \citenamefont {Milani}, \citenamefont {Spagnolo}, \citenamefont
  {Kaplan}, \citenamefont {Sciarrino},\ and\ \citenamefont
  {Kashefi}}]{polacchi2023multiclient}%
  \BibitemOpen
  \bibfield  {author} {\bibinfo {author} {\bibfnamefont {B.}~\bibnamefont
  {Polacchi}}, \bibinfo {author} {\bibfnamefont {D.}~\bibnamefont {Leichtle}},
  \bibinfo {author} {\bibfnamefont {L.}~\bibnamefont {Limongi}}, \bibinfo
  {author} {\bibfnamefont {G.}~\bibnamefont {Carvacho}}, \bibinfo {author}
  {\bibfnamefont {G.}~\bibnamefont {Milani}}, \bibinfo {author} {\bibfnamefont
  {N.}~\bibnamefont {Spagnolo}}, \bibinfo {author} {\bibfnamefont
  {M.}~\bibnamefont {Kaplan}}, \bibinfo {author} {\bibfnamefont
  {F.}~\bibnamefont {Sciarrino}},\ and\ \bibinfo {author} {\bibfnamefont
  {E.}~\bibnamefont {Kashefi}},\ }\href@noop {} {\bibinfo {title} {Multi-client
  distributed blind quantum computation with the qline architecture}} (\bibinfo
  {year} {2023}),\ \Eprint {https://arxiv.org/abs/2306.05195} {arXiv:2306.05195
  [quant-ph]} \BibitemShut {NoStop}%
\bibitem [{\citenamefont {McMahan}\ \emph {et~al.}(2017)\citenamefont
  {McMahan}, \citenamefont {Moore}, \citenamefont {Ramage}, \citenamefont
  {Hampson},\ and\ \citenamefont {y~Arcas}}]{mcmahan2017communication}%
  \BibitemOpen
  \bibfield  {author} {\bibinfo {author} {\bibfnamefont {B.}~\bibnamefont
  {McMahan}}, \bibinfo {author} {\bibfnamefont {E.}~\bibnamefont {Moore}},
  \bibinfo {author} {\bibfnamefont {D.}~\bibnamefont {Ramage}}, \bibinfo
  {author} {\bibfnamefont {S.}~\bibnamefont {Hampson}},\ and\ \bibinfo {author}
  {\bibfnamefont {B.~A.}\ \bibnamefont {y~Arcas}},\ }\bibfield  {title}
  {\bibinfo {title} {Communication-efficient learning of deep networks from
  decentralized data},\ }in\ \href@noop {} {\emph {\bibinfo {booktitle}
  {Artificial intelligence and statistics}}}\ (\bibinfo {organization} {PMLR},\
  \bibinfo {year} {2017})\ pp.\ \bibinfo {pages} {1273--1282}\BibitemShut
  {NoStop}%
\bibitem [{\citenamefont {Yang}\ \emph {et~al.}(2019)\citenamefont {Yang},
  \citenamefont {Liu}, \citenamefont {Chen},\ and\ \citenamefont
  {Tong}}]{yang2019federated}%
  \BibitemOpen
  \bibfield  {author} {\bibinfo {author} {\bibfnamefont {Q.}~\bibnamefont
  {Yang}}, \bibinfo {author} {\bibfnamefont {Y.}~\bibnamefont {Liu}}, \bibinfo
  {author} {\bibfnamefont {T.}~\bibnamefont {Chen}},\ and\ \bibinfo {author}
  {\bibfnamefont {Y.}~\bibnamefont {Tong}},\ }\bibfield  {title} {\bibinfo
  {title} {Federated machine learning: Concept and applications},\ }\href
  {https://doi.org/10.1145/3298981} {\bibfield  {journal} {\bibinfo  {journal}
  {ACM Transactions on Intelligent Systems and Technology}\ }\textbf {\bibinfo
  {volume} {10}},\ \bibinfo {pages} {1} (\bibinfo {year} {2019})}\BibitemShut
  {NoStop}%
\bibitem [{\citenamefont {Zhao}\ \emph {et~al.}(2020)\citenamefont {Zhao},
  \citenamefont {Mopuri},\ and\ \citenamefont {Bilen}}]{zhao2020idlg}%
  \BibitemOpen
  \bibfield  {author} {\bibinfo {author} {\bibfnamefont {B.}~\bibnamefont
  {Zhao}}, \bibinfo {author} {\bibfnamefont {K.~R.}\ \bibnamefont {Mopuri}},\
  and\ \bibinfo {author} {\bibfnamefont {H.}~\bibnamefont {Bilen}},\
  }\href@noop {} {\bibinfo {title} {idlg: Improved deep leakage from
  gradients}} (\bibinfo {year} {2020}),\ \Eprint
  {https://arxiv.org/abs/2001.02610} {arXiv:2001.02610 [cs.LG]} \BibitemShut
  {NoStop}%
\bibitem [{\citenamefont {Eloul}\ \emph {et~al.}(2022)\citenamefont {Eloul},
  \citenamefont {Silavong}, \citenamefont {Kamthe}, \citenamefont
  {Georgiadis},\ and\ \citenamefont {Moran}}]{eloul2022enhancing}%
  \BibitemOpen
  \bibfield  {author} {\bibinfo {author} {\bibfnamefont {S.}~\bibnamefont
  {Eloul}}, \bibinfo {author} {\bibfnamefont {F.}~\bibnamefont {Silavong}},
  \bibinfo {author} {\bibfnamefont {S.}~\bibnamefont {Kamthe}}, \bibinfo
  {author} {\bibfnamefont {A.}~\bibnamefont {Georgiadis}},\ and\ \bibinfo
  {author} {\bibfnamefont {S.~J.}\ \bibnamefont {Moran}},\ }\href@noop {}
  {\bibinfo {title} {Enhancing privacy against inversion attacks in federated
  learning by using mixing gradients strategies}} (\bibinfo {year} {2022}),\
  \Eprint {https://arxiv.org/abs/2204.12495} {arXiv:2204.12495 [cs.LG]}
  \BibitemShut {NoStop}%
\bibitem [{\citenamefont {Mothukuri}\ \emph {et~al.}(2021)\citenamefont
  {Mothukuri}, \citenamefont {Parizi}, \citenamefont {Pouriyeh}, \citenamefont
  {Huang}, \citenamefont {Dehghantanha},\ and\ \citenamefont
  {Srivastava}}]{mothukuri2021survey}%
  \BibitemOpen
  \bibfield  {author} {\bibinfo {author} {\bibfnamefont {V.}~\bibnamefont
  {Mothukuri}}, \bibinfo {author} {\bibfnamefont {R.~M.}\ \bibnamefont
  {Parizi}}, \bibinfo {author} {\bibfnamefont {S.}~\bibnamefont {Pouriyeh}},
  \bibinfo {author} {\bibfnamefont {Y.}~\bibnamefont {Huang}}, \bibinfo
  {author} {\bibfnamefont {A.}~\bibnamefont {Dehghantanha}},\ and\ \bibinfo
  {author} {\bibfnamefont {G.}~\bibnamefont {Srivastava}},\ }\bibfield  {title}
  {\bibinfo {title} {A survey on security and privacy of federated learning},\
  }\href {https://doi.org/10.1016/j.future.2020.10.007} {\bibfield  {journal}
  {\bibinfo  {journal} {Future Generation Computer Systems}\ }\textbf {\bibinfo
  {volume} {115}},\ \bibinfo {pages} {619} (\bibinfo {year}
  {2021})}\BibitemShut {NoStop}%
\bibitem [{\citenamefont {Zhu}\ \emph {et~al.}(2019)\citenamefont {Zhu},
  \citenamefont {Liu},\ and\ \citenamefont {Han}}]{Zhu19}%
  \BibitemOpen
  \bibfield  {author} {\bibinfo {author} {\bibfnamefont {L.}~\bibnamefont
  {Zhu}}, \bibinfo {author} {\bibfnamefont {Z.}~\bibnamefont {Liu}},\ and\
  \bibinfo {author} {\bibfnamefont {S.}~\bibnamefont {Han}},\ }\bibfield
  {title} {\bibinfo {title} {Deep leakage from gradients},\ }in\ \href
  {https://proceedings.neurips.cc/paper/2019/file/60a6c4002cc7b29142def8871531281a-Paper.pdf}
  {\emph {\bibinfo {booktitle} {Advances in Neural Information Processing
  Systems}}},\ Vol.~\bibinfo {volume} {32}\ (\bibinfo  {publisher} {Curran
  Associates, Inc.},\ \bibinfo {year} {2019})\BibitemShut {NoStop}%
\bibitem [{\citenamefont {Geiping}\ \emph
  {et~al.}(2020{\natexlab{a}})\citenamefont {Geiping}, \citenamefont
  {Bauermeister}, \citenamefont {Droge},\ and\ \citenamefont
  {Moeller}}]{geiping2020inverting}%
  \BibitemOpen
  \bibfield  {author} {\bibinfo {author} {\bibfnamefont {J.}~\bibnamefont
  {Geiping}}, \bibinfo {author} {\bibfnamefont {H.}~\bibnamefont
  {Bauermeister}}, \bibinfo {author} {\bibfnamefont {H.}~\bibnamefont
  {Droge}},\ and\ \bibinfo {author} {\bibfnamefont {M.}~\bibnamefont
  {Moeller}},\ }\href@noop {} {\bibinfo {title} {Inverting gradients -- how
  easy is it to break privacy in federated learning?}} (\bibinfo {year}
  {2020}{\natexlab{a}}),\ \Eprint {https://arxiv.org/abs/2003.14053}
  {arXiv:2003.14053 [cs.CV]} \BibitemShut {NoStop}%
\bibitem [{\citenamefont {Phong}\ \emph {et~al.}(2018)\citenamefont {Phong},
  \citenamefont {Aono}, \citenamefont {Hayashi}, \citenamefont {Wang},\ and\
  \citenamefont {Moriai}}]{aono2017privacy}%
  \BibitemOpen
  \bibfield  {author} {\bibinfo {author} {\bibfnamefont {L.~T.}\ \bibnamefont
  {Phong}}, \bibinfo {author} {\bibfnamefont {Y.}~\bibnamefont {Aono}},
  \bibinfo {author} {\bibfnamefont {T.}~\bibnamefont {Hayashi}}, \bibinfo
  {author} {\bibfnamefont {L.}~\bibnamefont {Wang}},\ and\ \bibinfo {author}
  {\bibfnamefont {S.}~\bibnamefont {Moriai}},\ }\bibfield  {title} {\bibinfo
  {title} {Privacy-preserving deep learning via additively homomorphic
  encryption},\ }\href {https://doi.org/10.1109/tifs.2017.2787987} {\bibfield
  {journal} {\bibinfo  {journal} {IEEE Transactions on Information Forensics
  and Security}\ }\textbf {\bibinfo {volume} {13}},\ \bibinfo {pages} {1333}
  (\bibinfo {year} {2018})}\BibitemShut {NoStop}%
\bibitem [{\citenamefont {Huang}\ \emph {et~al.}(2021)\citenamefont {Huang},
  \citenamefont {Gupta}, \citenamefont {Song}, \citenamefont {Li},\ and\
  \citenamefont {Arora}}]{huang2021evaluating}%
  \BibitemOpen
  \bibfield  {author} {\bibinfo {author} {\bibfnamefont {Y.}~\bibnamefont
  {Huang}}, \bibinfo {author} {\bibfnamefont {S.}~\bibnamefont {Gupta}},
  \bibinfo {author} {\bibfnamefont {Z.}~\bibnamefont {Song}}, \bibinfo {author}
  {\bibfnamefont {K.}~\bibnamefont {Li}},\ and\ \bibinfo {author}
  {\bibfnamefont {S.}~\bibnamefont {Arora}},\ }\bibfield  {title} {\bibinfo
  {title} {Evaluating gradient inversion attacks and defenses in federated
  learning},\ }in\ \href {https://openreview.net/forum?id=0CDKgyYaxC8} {\emph
  {\bibinfo {booktitle} {Advances in Neural Information Processing Systems}}},\
  \bibinfo {editor} {edited by\ \bibinfo {editor} {\bibfnamefont
  {A.}~\bibnamefont {Beygelzimer}}, \bibinfo {editor} {\bibfnamefont
  {Y.}~\bibnamefont {Dauphin}}, \bibinfo {editor} {\bibfnamefont
  {P.}~\bibnamefont {Liang}},\ and\ \bibinfo {editor} {\bibfnamefont {J.~W.}\
  \bibnamefont {Vaughan}}}\ (\bibinfo {year} {2021})\BibitemShut {NoStop}%
\bibitem [{\citenamefont {Du}\ \emph {et~al.}(2021)\citenamefont {Du},
  \citenamefont {Hsieh}, \citenamefont {Liu}, \citenamefont {Tao},\ and\
  \citenamefont {Liu}}]{Noise_PRR2021}%
  \BibitemOpen
  \bibfield  {author} {\bibinfo {author} {\bibfnamefont {Y.}~\bibnamefont
  {Du}}, \bibinfo {author} {\bibfnamefont {M.-H.}\ \bibnamefont {Hsieh}},
  \bibinfo {author} {\bibfnamefont {T.}~\bibnamefont {Liu}}, \bibinfo {author}
  {\bibfnamefont {D.}~\bibnamefont {Tao}},\ and\ \bibinfo {author}
  {\bibfnamefont {N.}~\bibnamefont {Liu}},\ }\bibfield  {title} {\bibinfo
  {title} {Quantum noise protects quantum classifiers against adversaries},\
  }\href {https://doi.org/10.1103/PhysRevResearch.3.023153} {\bibfield
  {journal} {\bibinfo  {journal} {Phys. Rev. Res.}\ }\textbf {\bibinfo {volume}
  {3}},\ \bibinfo {pages} {023153} (\bibinfo {year} {2021})}\BibitemShut
  {NoStop}%
\bibitem [{\citenamefont {{Li}}\ \emph {et~al.}(2021)\citenamefont {{Li}},
  \citenamefont {{Lu}},\ and\ \citenamefont {{Deng}}}]{Weikang2021}%
  \BibitemOpen
  \bibfield  {author} {\bibinfo {author} {\bibfnamefont {W.}~\bibnamefont
  {{Li}}}, \bibinfo {author} {\bibfnamefont {S.}~\bibnamefont {{Lu}}},\ and\
  \bibinfo {author} {\bibfnamefont {D.-L.}\ \bibnamefont {{Deng}}},\ }\bibfield
   {title} {\bibinfo {title} {{Quantum federated learning through blind quantum
  computing}},\ }\href {https://doi.org/10.1007/s11433-021-1753-3} {\bibfield
  {journal} {\bibinfo  {journal} {Science China Physics, Mechanics, and
  Astronomy}\ }\textbf {\bibinfo {volume} {64}},\ \bibinfo {eid} {100312}
  (\bibinfo {year} {2021})},\ \Eprint {https://arxiv.org/abs/2103.08403}
  {arXiv:2103.08403 [quant-ph]} \BibitemShut {NoStop}%
\bibitem [{\citenamefont {Ren}\ \emph {et~al.}(2023)\citenamefont {Ren},
  \citenamefont {Yu}, \citenamefont {Yan}, \citenamefont {Xu}, \citenamefont
  {Shen}, \citenamefont {Zhu}, \citenamefont {Niyato}, \citenamefont {Dong},\
  and\ \citenamefont {Kwek}}]{ren2023QFLreview}%
  \BibitemOpen
  \bibfield  {author} {\bibinfo {author} {\bibfnamefont {C.}~\bibnamefont
  {Ren}}, \bibinfo {author} {\bibfnamefont {H.}~\bibnamefont {Yu}}, \bibinfo
  {author} {\bibfnamefont {R.}~\bibnamefont {Yan}}, \bibinfo {author}
  {\bibfnamefont {M.}~\bibnamefont {Xu}}, \bibinfo {author} {\bibfnamefont
  {Y.}~\bibnamefont {Shen}}, \bibinfo {author} {\bibfnamefont {H.}~\bibnamefont
  {Zhu}}, \bibinfo {author} {\bibfnamefont {D.}~\bibnamefont {Niyato}},
  \bibinfo {author} {\bibfnamefont {Z.~Y.}\ \bibnamefont {Dong}},\ and\
  \bibinfo {author} {\bibfnamefont {L.~C.}\ \bibnamefont {Kwek}},\ }\href@noop
  {} {\bibinfo {title} {Towards quantum federated learning}} (\bibinfo {year}
  {2023}),\ \Eprint {https://arxiv.org/abs/2306.09912} {arXiv:2306.09912
  [cs.LG]} \BibitemShut {NoStop}%
\bibitem [{\citenamefont {Chen}\ and\ \citenamefont {Yoo}(2021)}]{chen2021}%
  \BibitemOpen
  \bibfield  {author} {\bibinfo {author} {\bibfnamefont {S.~Y.-C.}\
  \bibnamefont {Chen}}\ and\ \bibinfo {author} {\bibfnamefont {S.}~\bibnamefont
  {Yoo}},\ }\bibfield  {title} {\bibinfo {title} {Federated quantum machine
  learning},\ }\bibfield  {journal} {\bibinfo  {journal} {Entropy}\ }\textbf
  {\bibinfo {volume} {23}},\ \href {https://doi.org/10.3390/e23040460}
  {10.3390/e23040460} (\bibinfo {year} {2021})\BibitemShut {NoStop}%
\bibitem [{\citenamefont {Huang}\ \emph {et~al.}(2022)\citenamefont {Huang},
  \citenamefont {Tan},\ and\ \citenamefont {Xu}}]{Huang2022}%
  \BibitemOpen
  \bibfield  {author} {\bibinfo {author} {\bibfnamefont {R.}~\bibnamefont
  {Huang}}, \bibinfo {author} {\bibfnamefont {X.}~\bibnamefont {Tan}},\ and\
  \bibinfo {author} {\bibfnamefont {Q.}~\bibnamefont {Xu}},\ }\bibfield
  {title} {\bibinfo {title} {Quantum federated learning with decentralized
  data},\ }\href {https://doi.org/10.1109/JSTQE.2022.3170150} {\bibfield
  {journal} {\bibinfo  {journal} {IEEE Journal of Selected Topics in Quantum
  Electronics}\ }\textbf {\bibinfo {volume} {28}},\ \bibinfo {pages} {1}
  (\bibinfo {year} {2022})}\BibitemShut {NoStop}%
\bibitem [{\citenamefont {Chu}\ \emph {et~al.}(2023)\citenamefont {Chu},
  \citenamefont {Jiang},\ and\ \citenamefont {Chen}}]{chu2023cryptoqfl}%
  \BibitemOpen
  \bibfield  {author} {\bibinfo {author} {\bibfnamefont {C.}~\bibnamefont
  {Chu}}, \bibinfo {author} {\bibfnamefont {L.}~\bibnamefont {Jiang}},\ and\
  \bibinfo {author} {\bibfnamefont {F.}~\bibnamefont {Chen}},\ }\href@noop {}
  {\bibinfo {title} {Cryptoqfl: Quantum federated learning on encrypted data}}
  (\bibinfo {year} {2023}),\ \Eprint {https://arxiv.org/abs/2307.07012}
  {arXiv:2307.07012 [quant-ph]} \BibitemShut {NoStop}%
\bibitem [{\citenamefont {Gisin}\ \emph {et~al.}(2002)\citenamefont {Gisin},
  \citenamefont {Ribordy}, \citenamefont {Tittel},\ and\ \citenamefont
  {Zbinden}}]{QCryptoRevModPhys.74.145}%
  \BibitemOpen
  \bibfield  {author} {\bibinfo {author} {\bibfnamefont {N.}~\bibnamefont
  {Gisin}}, \bibinfo {author} {\bibfnamefont {G.}~\bibnamefont {Ribordy}},
  \bibinfo {author} {\bibfnamefont {W.}~\bibnamefont {Tittel}},\ and\ \bibinfo
  {author} {\bibfnamefont {H.}~\bibnamefont {Zbinden}},\ }\bibfield  {title}
  {\bibinfo {title} {Quantum cryptography},\ }\href
  {https://doi.org/10.1103/RevModPhys.74.145} {\bibfield  {journal} {\bibinfo
  {journal} {Rev. Mod. Phys.}\ }\textbf {\bibinfo {volume} {74}},\ \bibinfo
  {pages} {145} (\bibinfo {year} {2002})}\BibitemShut {NoStop}%
\bibitem [{\citenamefont {Bonawitz}\ \emph {et~al.}(2016)\citenamefont
  {Bonawitz}, \citenamefont {Ivanov}, \citenamefont {Kreuter}, \citenamefont
  {Marcedone}, \citenamefont {McMahan}, \citenamefont {Patel}, \citenamefont
  {Ramage}, \citenamefont {Segal},\ and\ \citenamefont
  {Seth}}]{bonawitz2016practical}%
  \BibitemOpen
  \bibfield  {author} {\bibinfo {author} {\bibfnamefont {K.}~\bibnamefont
  {Bonawitz}}, \bibinfo {author} {\bibfnamefont {V.}~\bibnamefont {Ivanov}},
  \bibinfo {author} {\bibfnamefont {B.}~\bibnamefont {Kreuter}}, \bibinfo
  {author} {\bibfnamefont {A.}~\bibnamefont {Marcedone}}, \bibinfo {author}
  {\bibfnamefont {H.~B.}\ \bibnamefont {McMahan}}, \bibinfo {author}
  {\bibfnamefont {S.}~\bibnamefont {Patel}}, \bibinfo {author} {\bibfnamefont
  {D.}~\bibnamefont {Ramage}}, \bibinfo {author} {\bibfnamefont
  {A.}~\bibnamefont {Segal}},\ and\ \bibinfo {author} {\bibfnamefont
  {K.}~\bibnamefont {Seth}},\ }\href@noop {} {\bibinfo {title} {Practical
  secure aggregation for federated learning on user-held data}} (\bibinfo
  {year} {2016}),\ \Eprint {https://arxiv.org/abs/1611.04482} {arXiv:1611.04482
  [cs.CR]} \BibitemShut {NoStop}%
\bibitem [{\citenamefont {Geiping}\ \emph
  {et~al.}(2020{\natexlab{b}})\citenamefont {Geiping}, \citenamefont
  {Bauermeister}, \citenamefont {Dr{\"{o}}ge},\ and\ \citenamefont
  {Moeller}}]{GeipingBD020}%
  \BibitemOpen
  \bibfield  {author} {\bibinfo {author} {\bibfnamefont {J.}~\bibnamefont
  {Geiping}}, \bibinfo {author} {\bibfnamefont {H.}~\bibnamefont
  {Bauermeister}}, \bibinfo {author} {\bibfnamefont {H.}~\bibnamefont
  {Dr{\"{o}}ge}},\ and\ \bibinfo {author} {\bibfnamefont {M.}~\bibnamefont
  {Moeller}},\ }\bibfield  {title} {\bibinfo {title} {{Inverting Gradients -
  How easy is it to break privacy in federated learning?}},\ }in\ \href
  {https://proceedings.neurips.cc/paper/2020/hash/c4ede56bbd98819ae6112b20ac6bf145-Abstract.html}
  {\emph {\bibinfo {booktitle} {Advances in Neural Information Processing
  Systems 33: Annual Conference on Neural Information Processing Systems 2020,
  NeurIPS 2020, December 6-12, 2020, virtual}}}\ (\bibinfo {year}
  {2020})\BibitemShut {NoStop}%
\bibitem [{\citenamefont {Yin}\ \emph {et~al.}(2021)\citenamefont {Yin},
  \citenamefont {Mallya}, \citenamefont {Vahdat}, \citenamefont {Alvarez},
  \citenamefont {Kautz},\ and\ \citenamefont {Molchanov}}]{Yin21}%
  \BibitemOpen
  \bibfield  {author} {\bibinfo {author} {\bibfnamefont {H.}~\bibnamefont
  {Yin}}, \bibinfo {author} {\bibfnamefont {A.}~\bibnamefont {Mallya}},
  \bibinfo {author} {\bibfnamefont {A.}~\bibnamefont {Vahdat}}, \bibinfo
  {author} {\bibfnamefont {J.}~\bibnamefont {Alvarez}}, \bibinfo {author}
  {\bibfnamefont {J.}~\bibnamefont {Kautz}},\ and\ \bibinfo {author}
  {\bibfnamefont {P.}~\bibnamefont {Molchanov}},\ }\bibfield  {title} {\bibinfo
  {title} {{See through Gradients: Image Batch Recovery via GradInversion}},\
  }in\ \href {https://doi.org/10.1109/CVPR46437.2021.01607} {\emph {\bibinfo
  {booktitle} {2021 IEEE/CVF Conference on Computer Vision and Pattern
  Recognition (CVPR)}}}\ (\bibinfo {year} {2021})\ pp.\ \bibinfo {pages}
  {16332--16341}\BibitemShut {NoStop}%
\bibitem [{\citenamefont {Zhu}\ and\ \citenamefont {Blaschko}(2021)}]{ZhuB21}%
  \BibitemOpen
  \bibfield  {author} {\bibinfo {author} {\bibfnamefont {J.}~\bibnamefont
  {Zhu}}\ and\ \bibinfo {author} {\bibfnamefont {M.~B.}\ \bibnamefont
  {Blaschko}},\ }\bibfield  {title} {\bibinfo {title} {{R-GAP:} recursive
  gradient attack on privacy},\ }in\ \href
  {https://openreview.net/forum?id=RSU17UoKfJF} {\emph {\bibinfo {booktitle}
  {9th International Conference on Learning Representations, {ICLR} 2021,
  Virtual Event, Austria, May 3-7, 2021}}}\ (\bibinfo  {publisher}
  {OpenReview.net},\ \bibinfo {year} {2021})\BibitemShut {NoStop}%
\bibitem [{\citenamefont {Benedetti}\ \emph {et~al.}(2019)\citenamefont
  {Benedetti}, \citenamefont {Lloyd}, \citenamefont {Sack},\ and\ \citenamefont
  {Fiorentini}}]{benedetti2019parameterized}%
  \BibitemOpen
  \bibfield  {author} {\bibinfo {author} {\bibfnamefont {M.}~\bibnamefont
  {Benedetti}}, \bibinfo {author} {\bibfnamefont {E.}~\bibnamefont {Lloyd}},
  \bibinfo {author} {\bibfnamefont {S.}~\bibnamefont {Sack}},\ and\ \bibinfo
  {author} {\bibfnamefont {M.}~\bibnamefont {Fiorentini}},\ }\bibfield  {title}
  {\bibinfo {title} {Parameterized quantum circuits as machine learning
  models},\ }\href {https://doi.org/10.1088/2058-9565/ab4eb5} {\bibfield
  {journal} {\bibinfo  {journal} {Quantum Science and Technology}\ }\textbf
  {\bibinfo {volume} {4}},\ \bibinfo {pages} {043001} (\bibinfo {year}
  {2019})}\BibitemShut {NoStop}%
\bibitem [{\citenamefont {Buhrman}\ \emph {et~al.}(2001)\citenamefont
  {Buhrman}, \citenamefont {Cleve}, \citenamefont {Watrous},\ and\
  \citenamefont {De~Wolf}}]{buhrman2001quantum}%
  \BibitemOpen
  \bibfield  {author} {\bibinfo {author} {\bibfnamefont {H.}~\bibnamefont
  {Buhrman}}, \bibinfo {author} {\bibfnamefont {R.}~\bibnamefont {Cleve}},
  \bibinfo {author} {\bibfnamefont {J.}~\bibnamefont {Watrous}},\ and\ \bibinfo
  {author} {\bibfnamefont {R.}~\bibnamefont {De~Wolf}},\ }\bibfield  {title}
  {\bibinfo {title} {Quantum fingerprinting},\ }\href@noop {} {\bibfield
  {journal} {\bibinfo  {journal} {Physical Review Letters}\ }\textbf {\bibinfo
  {volume} {87}},\ \bibinfo {pages} {167902} (\bibinfo {year}
  {2001})}\BibitemShut {NoStop}%
\bibitem [{IEE()}]{IEEEstandard2019}%
  \BibitemOpen
  \href {https://doi.org/10.1109/ieeestd.2019.8766229} {\bibinfo {title}
  {{IEEE} standard for floating-point arithmetic}}\BibitemShut {NoStop}%
\bibitem [{\citenamefont {Fanizza}\ \emph {et~al.}(2020)\citenamefont
  {Fanizza}, \citenamefont {Rosati}, \citenamefont {Skotiniotis}, \citenamefont
  {Calsamiglia},\ and\ \citenamefont {Giovannetti}}]{PhysRevLett.124.060503}%
  \BibitemOpen
  \bibfield  {author} {\bibinfo {author} {\bibfnamefont {M.}~\bibnamefont
  {Fanizza}}, \bibinfo {author} {\bibfnamefont {M.}~\bibnamefont {Rosati}},
  \bibinfo {author} {\bibfnamefont {M.}~\bibnamefont {Skotiniotis}}, \bibinfo
  {author} {\bibfnamefont {J.}~\bibnamefont {Calsamiglia}},\ and\ \bibinfo
  {author} {\bibfnamefont {V.}~\bibnamefont {Giovannetti}},\ }\bibfield
  {title} {\bibinfo {title} {Beyond the swap test: Optimal estimation of
  quantum state overlap},\ }\href
  {https://doi.org/10.1103/PhysRevLett.124.060503} {\bibfield  {journal}
  {\bibinfo  {journal} {Phys. Rev. Lett.}\ }\textbf {\bibinfo {volume} {124}},\
  \bibinfo {pages} {060503} (\bibinfo {year} {2020})}\BibitemShut {NoStop}%
\bibitem [{\citenamefont {Zhang}\ \emph {et~al.}(2023)\citenamefont {Zhang},
  \citenamefont {Zhang}, \citenamefont {Zhang}, \citenamefont {Fan},
  \citenamefont {Zeng},\ and\ \citenamefont {Yang}}]{zhang2023federated}%
  \BibitemOpen
  \bibfield  {author} {\bibinfo {author} {\bibfnamefont {Y.}~\bibnamefont
  {Zhang}}, \bibinfo {author} {\bibfnamefont {C.}~\bibnamefont {Zhang}},
  \bibinfo {author} {\bibfnamefont {C.}~\bibnamefont {Zhang}}, \bibinfo
  {author} {\bibfnamefont {L.}~\bibnamefont {Fan}}, \bibinfo {author}
  {\bibfnamefont {B.}~\bibnamefont {Zeng}},\ and\ \bibinfo {author}
  {\bibfnamefont {Q.}~\bibnamefont {Yang}},\ }\href@noop {} {\bibinfo {title}
  {Federated learning with quantum secure aggregation}} (\bibinfo {year}
  {2023}),\ \Eprint {https://arxiv.org/abs/2207.07444} {arXiv:2207.07444
  [quant-ph]} \BibitemShut {NoStop}%
\bibitem [{\citenamefont {Nielsen}\ and\ \citenamefont
  {Chuang}(2010)}]{nielsen2010quantum}%
  \BibitemOpen
  \bibfield  {author} {\bibinfo {author} {\bibfnamefont {M.~A.}\ \bibnamefont
  {Nielsen}}\ and\ \bibinfo {author} {\bibfnamefont {I.~L.}\ \bibnamefont
  {Chuang}},\ }\href@noop {} {\emph {\bibinfo {title} {Quantum computation and
  quantum information}}}\ (\bibinfo  {publisher} {Cambridge university press},\
  \bibinfo {year} {2010})\BibitemShut {NoStop}%
\bibitem [{\citenamefont {hua Shi}\ \emph {et~al.}(2016)\citenamefont {hua
  Shi}, \citenamefont {Mu}, \citenamefont {Zhong}, \citenamefont {Cui},\ and\
  \citenamefont {Zhang}}]{Shi2016}%
  \BibitemOpen
  \bibfield  {author} {\bibinfo {author} {\bibfnamefont {R.}~\bibnamefont {hua
  Shi}}, \bibinfo {author} {\bibfnamefont {Y.}~\bibnamefont {Mu}}, \bibinfo
  {author} {\bibfnamefont {H.}~\bibnamefont {Zhong}}, \bibinfo {author}
  {\bibfnamefont {J.}~\bibnamefont {Cui}},\ and\ \bibinfo {author}
  {\bibfnamefont {S.}~\bibnamefont {Zhang}},\ }\bibfield  {title} {\bibinfo
  {title} {Secure multiparty quantum computation for summation and
  multiplication},\ }\bibfield  {journal} {\bibinfo  {journal} {Scientific
  Reports}\ }\textbf {\bibinfo {volume} {6}},\ \href
  {https://doi.org/10.1038/srep19655} {10.1038/srep19655} (\bibinfo {year}
  {2016})\BibitemShut {NoStop}%
\bibitem [{\citenamefont {Lo}\ \emph {et~al.}(2005)\citenamefont {Lo},
  \citenamefont {Ma},\ and\ \citenamefont {Chen}}]{DecoyPhysRevLett.94.230504}%
  \BibitemOpen
  \bibfield  {author} {\bibinfo {author} {\bibfnamefont {H.-K.}\ \bibnamefont
  {Lo}}, \bibinfo {author} {\bibfnamefont {X.}~\bibnamefont {Ma}},\ and\
  \bibinfo {author} {\bibfnamefont {K.}~\bibnamefont {Chen}},\ }\bibfield
  {title} {\bibinfo {title} {Decoy state quantum key distribution},\ }\href
  {https://doi.org/10.1103/PhysRevLett.94.230504} {\bibfield  {journal}
  {\bibinfo  {journal} {Phys. Rev. Lett.}\ }\textbf {\bibinfo {volume} {94}},\
  \bibinfo {pages} {230504} (\bibinfo {year} {2005})}\BibitemShut {NoStop}%
\bibitem [{\citenamefont {Sheng}\ and\ \citenamefont
  {Zhou}(2017)}]{Sheng2017SB}%
  \BibitemOpen
  \bibfield  {author} {\bibinfo {author} {\bibfnamefont {Y.-B.}\ \bibnamefont
  {Sheng}}\ and\ \bibinfo {author} {\bibfnamefont {L.}~\bibnamefont {Zhou}},\
  }\bibfield  {title} {\bibinfo {title} {Distributed secure quantum machine
  learning},\ }\href {https://doi.org/10.1016/j.scib.2017.06.007} {\bibfield
  {journal} {\bibinfo  {journal} {Science Bulletin}\ }\textbf {\bibinfo
  {volume} {62}},\ \bibinfo {pages} {1025} (\bibinfo {year}
  {2017})}\BibitemShut {NoStop}%
\bibitem [{\citenamefont {Zhang}\ \emph {et~al.}(2022)\citenamefont {Zhang},
  \citenamefont {Gu}, \citenamefont {Fan}, \citenamefont {Chen},\ and\
  \citenamefont {Yang}}]{zhang2022free}%
  \BibitemOpen
  \bibfield  {author} {\bibinfo {author} {\bibfnamefont {X.}~\bibnamefont
  {Zhang}}, \bibinfo {author} {\bibfnamefont {H.}~\bibnamefont {Gu}}, \bibinfo
  {author} {\bibfnamefont {L.}~\bibnamefont {Fan}}, \bibinfo {author}
  {\bibfnamefont {K.}~\bibnamefont {Chen}},\ and\ \bibinfo {author}
  {\bibfnamefont {Q.}~\bibnamefont {Yang}},\ }\href@noop {} {\bibinfo {title}
  {No free lunch theorem for security and utility in federated learning}}
  (\bibinfo {year} {2022}),\ \Eprint {https://arxiv.org/abs/2203.05816}
  {arXiv:2203.05816 [cs.LG]} \BibitemShut {NoStop}%
\bibitem [{\citenamefont {Dulek}\ \emph {et~al.}(2020)\citenamefont {Dulek},
  \citenamefont {Grilo}, \citenamefont {Jeffery}, \citenamefont {Majenz},\ and\
  \citenamefont {Schaffner}}]{Dulek2020Dihonest_majority}%
  \BibitemOpen
  \bibfield  {author} {\bibinfo {author} {\bibfnamefont {Y.}~\bibnamefont
  {Dulek}}, \bibinfo {author} {\bibfnamefont {A.~B.}\ \bibnamefont {Grilo}},
  \bibinfo {author} {\bibfnamefont {S.}~\bibnamefont {Jeffery}}, \bibinfo
  {author} {\bibfnamefont {C.}~\bibnamefont {Majenz}},\ and\ \bibinfo {author}
  {\bibfnamefont {C.}~\bibnamefont {Schaffner}},\ }\bibinfo {title} {Secure
  multi-party quantum computation with a dishonest majority},\ in\ \href
  {https://doi.org/10.1007/978-3-030-45727-3_25} {\emph {\bibinfo {booktitle}
  {Lecture Notes in Computer Science}}}\ (\bibinfo  {publisher} {Springer
  International Publishing},\ \bibinfo {year} {2020})\BibitemShut {NoStop}%
\bibitem [{\citenamefont {Xia}\ \emph {et~al.}(2021)\citenamefont {Xia},
  \citenamefont {Tao},\ and\ \citenamefont {Li}}]{Xia2021ByzantineAttack}%
  \BibitemOpen
  \bibfield  {author} {\bibinfo {author} {\bibfnamefont {Q.}~\bibnamefont
  {Xia}}, \bibinfo {author} {\bibfnamefont {Z.}~\bibnamefont {Tao}},\ and\
  \bibinfo {author} {\bibfnamefont {Q.}~\bibnamefont {Li}},\ }\bibfield
  {title} {\bibinfo {title} {Defending against byzantine attacks in quantum
  federated learning},\ }in\ \href
  {https://doi.org/10.1109/msn53354.2021.00035} {\emph {\bibinfo {booktitle}
  {2021 17th International Conference on Mobility, Sensing and Networking
  (MSN)}}}\ (\bibinfo  {publisher} {IEEE},\ \bibinfo {year} {2021})\BibitemShut
  {NoStop}%
\bibitem [{\citenamefont {Arapinis}\ \emph {et~al.}(2021)\citenamefont
  {Arapinis}, \citenamefont {Lamprou}, \citenamefont {Kashefi},\ and\
  \citenamefont {Pappa}}]{Arapinis_2021E-voting}%
  \BibitemOpen
  \bibfield  {author} {\bibinfo {author} {\bibfnamefont {M.}~\bibnamefont
  {Arapinis}}, \bibinfo {author} {\bibfnamefont {N.}~\bibnamefont {Lamprou}},
  \bibinfo {author} {\bibfnamefont {E.}~\bibnamefont {Kashefi}},\ and\ \bibinfo
  {author} {\bibfnamefont {A.}~\bibnamefont {Pappa}},\ }\bibfield  {title}
  {\bibinfo {title} {Definitions and security of quantum electronic voting},\
  }\href {https://doi.org/10.1145/3450144} {\bibfield  {journal} {\bibinfo
  {journal} {ACM Transactions on Quantum Computing}\ }\textbf {\bibinfo
  {volume} {2}},\ \bibinfo {pages} {1} (\bibinfo {year} {2021})}\BibitemShut
  {NoStop}%
\bibitem [{\citenamefont {Centrone}\ \emph {et~al.}(2022)\citenamefont
  {Centrone}, \citenamefont {Diamanti},\ and\ \citenamefont
  {Kerenidis}}]{Centrone_2022E-voting}%
  \BibitemOpen
  \bibfield  {author} {\bibinfo {author} {\bibfnamefont {F.}~\bibnamefont
  {Centrone}}, \bibinfo {author} {\bibfnamefont {E.}~\bibnamefont {Diamanti}},\
  and\ \bibinfo {author} {\bibfnamefont {I.}~\bibnamefont {Kerenidis}},\
  }\bibfield  {title} {\bibinfo {title} {Quantum protocol for electronic voting
  without election authorities},\ }\bibfield  {journal} {\bibinfo  {journal}
  {Physical Review Applied}\ }\textbf {\bibinfo {volume} {18}},\ \href
  {https://doi.org/10.1103/physrevapplied.18.014005}
  {10.1103/physrevapplied.18.014005} (\bibinfo {year} {2022})\BibitemShut
  {NoStop}%
\bibitem [{\citenamefont {Gu\'erin}\ \emph {et~al.}(2016)\citenamefont
  {Gu\'erin}, \citenamefont {Feix}, \citenamefont {Ara\'ujo},\ and\
  \citenamefont {Brukner}}]{PhysRevLett.117.100502}%
  \BibitemOpen
  \bibfield  {author} {\bibinfo {author} {\bibfnamefont {P.~A.}\ \bibnamefont
  {Gu\'erin}}, \bibinfo {author} {\bibfnamefont {A.}~\bibnamefont {Feix}},
  \bibinfo {author} {\bibfnamefont {M.}~\bibnamefont {Ara\'ujo}},\ and\
  \bibinfo {author} {\bibfnamefont {i.~c.~v.}\ \bibnamefont {Brukner}},\
  }\bibfield  {title} {\bibinfo {title} {Exponential communication complexity
  advantage from quantum superposition of the direction of communication},\
  }\href {https://doi.org/10.1103/PhysRevLett.117.100502} {\bibfield  {journal}
  {\bibinfo  {journal} {Phys. Rev. Lett.}\ }\textbf {\bibinfo {volume} {117}},\
  \bibinfo {pages} {100502} (\bibinfo {year} {2016})}\BibitemShut {NoStop}%
\bibitem [{Note1()}]{Note1}%
  \BibitemOpen
  \bibinfo {note} {We represent the $y$ as $y_1\protect \cdots y_{|\protect
  \mathcal {C}|}$ such that $y = c_i$ implies $y_i = 1$ and rest being zero.
  This gives us $|\protect \mathcal {C}|$ number of finite output classes. As
  an example, for $|\protect \mathcal {C}| = 3$, the set $\protect \mathcal {C}
  = \protect \{a_1 : 100, a_2 : 010, a_3 : 001\protect \}$.}\BibitemShut
  {Stop}%
\bibitem [{\citenamefont {Rivest}\ \emph {et~al.}(1978)\citenamefont {Rivest},
  \citenamefont {Adleman}, \citenamefont {Dertouzos} \emph
  {et~al.}}]{rivest1978data}%
  \BibitemOpen
  \bibfield  {author} {\bibinfo {author} {\bibfnamefont {R.~L.}\ \bibnamefont
  {Rivest}}, \bibinfo {author} {\bibfnamefont {L.}~\bibnamefont {Adleman}},
  \bibinfo {author} {\bibfnamefont {M.~L.}\ \bibnamefont {Dertouzos}}, \emph
  {et~al.},\ }\bibfield  {title} {\bibinfo {title} {On data banks and privacy
  homomorphisms},\ }\href@noop {} {\bibfield  {journal} {\bibinfo  {journal}
  {Foundations of secure computation}\ }\textbf {\bibinfo {volume} {4}},\
  \bibinfo {pages} {169} (\bibinfo {year} {1978})}\BibitemShut {NoStop}%
\bibitem [{\citenamefont {Zhang}\ \emph {et~al.}(2020)\citenamefont {Zhang},
  \citenamefont {Li}, \citenamefont {Xia}, \citenamefont {Wang}, \citenamefont
  {Yan},\ and\ \citenamefont {Liu}}]{10.5555/3489146.3489179}%
  \BibitemOpen
  \bibfield  {author} {\bibinfo {author} {\bibfnamefont {C.}~\bibnamefont
  {Zhang}}, \bibinfo {author} {\bibfnamefont {S.}~\bibnamefont {Li}}, \bibinfo
  {author} {\bibfnamefont {J.}~\bibnamefont {Xia}}, \bibinfo {author}
  {\bibfnamefont {W.}~\bibnamefont {Wang}}, \bibinfo {author} {\bibfnamefont
  {F.}~\bibnamefont {Yan}},\ and\ \bibinfo {author} {\bibfnamefont
  {Y.}~\bibnamefont {Liu}},\ }\bibfield  {title} {\bibinfo {title} {Batchcrypt:
  Efficient homomorphic encryption for cross-silo federated learning},\ }in\
  \href@noop {} {\emph {\bibinfo {booktitle} {Proceedings of the 2020 USENIX
  Conference on Usenix Annual Technical Conference}}}\ (\bibinfo  {publisher}
  {USENIX Association},\ \bibinfo {address} {USA},\ \bibinfo {year}
  {2020})\BibitemShut {NoStop}%
\bibitem [{\citenamefont {Abadi}\ \emph {et~al.}(2016)\citenamefont {Abadi},
  \citenamefont {Chu}, \citenamefont {Goodfellow}, \citenamefont {McMahan},
  \citenamefont {Mironov}, \citenamefont {Talwar},\ and\ \citenamefont
  {Zhang}}]{Abadi2016DP}%
  \BibitemOpen
  \bibfield  {author} {\bibinfo {author} {\bibfnamefont {M.}~\bibnamefont
  {Abadi}}, \bibinfo {author} {\bibfnamefont {A.}~\bibnamefont {Chu}}, \bibinfo
  {author} {\bibfnamefont {I.}~\bibnamefont {Goodfellow}}, \bibinfo {author}
  {\bibfnamefont {H.~B.}\ \bibnamefont {McMahan}}, \bibinfo {author}
  {\bibfnamefont {I.}~\bibnamefont {Mironov}}, \bibinfo {author} {\bibfnamefont
  {K.}~\bibnamefont {Talwar}},\ and\ \bibinfo {author} {\bibfnamefont
  {L.}~\bibnamefont {Zhang}},\ }\bibfield  {title} {\bibinfo {title} {Deep
  learning with differential privacy},\ }in\ \href
  {https://doi.org/10.1145/2976749.2978318} {\emph {\bibinfo {booktitle}
  {Proceedings of the 2016 ACM SIGSAC Conference on Computer and Communications
  Security}}},\ \bibinfo {series and number} {CCS16}\ (\bibinfo  {publisher}
  {ACM},\ \bibinfo {year} {2016})\BibitemShut {NoStop}%
\bibitem [{\citenamefont {Dwork}\ \emph {et~al.}(2006)\citenamefont {Dwork},
  \citenamefont {McSherry}, \citenamefont {Nissim},\ and\ \citenamefont
  {Smith}}]{Dwork2006}%
  \BibitemOpen
  \bibfield  {author} {\bibinfo {author} {\bibfnamefont {C.}~\bibnamefont
  {Dwork}}, \bibinfo {author} {\bibfnamefont {F.}~\bibnamefont {McSherry}},
  \bibinfo {author} {\bibfnamefont {K.}~\bibnamefont {Nissim}},\ and\ \bibinfo
  {author} {\bibfnamefont {A.}~\bibnamefont {Smith}},\ }\bibinfo {title}
  {Calibrating noise to sensitivity in private data analysis},\ in\ \href
  {https://doi.org/10.1007/11681878_14} {\emph {\bibinfo {booktitle} {Lecture
  Notes in Computer Science}}}\ (\bibinfo  {publisher} {Springer Berlin
  Heidelberg},\ \bibinfo {year} {2006})\ pp.\ \bibinfo {pages}
  {265--284}\BibitemShut {NoStop}%
\bibitem [{\citenamefont {Dwork}(2006)}]{dwork2006differential}%
  \BibitemOpen
  \bibfield  {author} {\bibinfo {author} {\bibfnamefont {C.}~\bibnamefont
  {Dwork}},\ }\bibfield  {title} {\bibinfo {title} {Differential privacy},\
  }in\ \href@noop {} {\emph {\bibinfo {booktitle} {International colloquium on
  automata, languages, and programming}}}\ (\bibinfo {organization}
  {Springer},\ \bibinfo {year} {2006})\ pp.\ \bibinfo {pages}
  {1--12}\BibitemShut {NoStop}%
\bibitem [{\citenamefont {Truex}\ \emph {et~al.}(2020)\citenamefont {Truex},
  \citenamefont {Liu}, \citenamefont {Chow}, \citenamefont {Gursoy},\ and\
  \citenamefont {Wei}}]{truex2020ldpfed}%
  \BibitemOpen
  \bibfield  {author} {\bibinfo {author} {\bibfnamefont {S.}~\bibnamefont
  {Truex}}, \bibinfo {author} {\bibfnamefont {L.}~\bibnamefont {Liu}}, \bibinfo
  {author} {\bibfnamefont {K.-H.}\ \bibnamefont {Chow}}, \bibinfo {author}
  {\bibfnamefont {M.~E.}\ \bibnamefont {Gursoy}},\ and\ \bibinfo {author}
  {\bibfnamefont {W.}~\bibnamefont {Wei}},\ }\href@noop {} {\bibinfo {title}
  {Ldp-fed: Federated learning with local differential privacy}} (\bibinfo
  {year} {2020}),\ \Eprint {https://arxiv.org/abs/2006.03637} {arXiv:2006.03637
  [cs.LG]} \BibitemShut {NoStop}%
\bibitem [{\citenamefont {Shamir}(1979)}]{Shamir1979}%
  \BibitemOpen
  \bibfield  {author} {\bibinfo {author} {\bibfnamefont {A.}~\bibnamefont
  {Shamir}},\ }\bibfield  {title} {\bibinfo {title} {How to share a secret},\
  }\href {https://doi.org/10.1145/359168.359176} {\bibfield  {journal}
  {\bibinfo  {journal} {Communications of the ACM}\ }\textbf {\bibinfo {volume}
  {22}},\ \bibinfo {pages} {612} (\bibinfo {year} {1979})}\BibitemShut
  {NoStop}%
\end{thebibliography}%

\section*{Disclaimer}
This paper was prepared for informational purposes by the Global Technology Applied Research center of JPMorgan Chase \& Co. This paper is not a product of the Research Department of JPMorgan Chase \& Co. or its affiliates. Neither JPMorgan Chase \& Co. nor any of its affiliates makes any explicit or implied representation or warranty and none of them accept any liability in connection with this paper, including, without limitation, with respect to the completeness, accuracy, or reliability of the information contained herein and the potential legal, compliance, tax, or accounting effects thereof. This document is not intended as investment research or investment advice, or as a recommendation, offer, or solicitation for the purchase or sale of any security, financial instrument, financial product or service, or to be used in any way for evaluating the merits of participating in any transaction. Zhixin Song's contributions were made as part of his internship at Global Technology Applied Research in JPMorgan Chase.
\newpage
\clearpage

\appendix

\section{Gradient-inversion attack in federated learning} \label{app:gradient_inversion}

In this section, we expand upon the argument presented in Sec.~\ref{sec:data_leakage}, namely a major data leakage issue with the standard classical federated learning approach.
This vulnerability arises from the susceptibility to a gradient inversion attack, wherein an \emph{honest-but-curious} server can successfully reconstruct the original data from the received gradients~\cite{eloul2022enhancing, zhao2020idlg}. Here we analyze an example case by looking at the vulnerability of recovering averaged data information where the model considered for training is the fully connected neural network layer. To simplify the analysis, we consider a dense linear layer containing $\mathbf{x} = x_1 \cdots x_n$ as input and $y \in \mathcal{C}$ as output~\footnote{We represent the $y$ as $y_1\cdots y_{|\mathcal{C}|}$ such that $y = c_i$ implies $y_i = 1$ and rest being zero. This gives us $|\mathcal{C}|$ number of finite output classes. As an example, for $|\mathcal{C}| = 3$, the set $\mathcal{C} = \{a_1 : 100, a_2 : 010, a_3 : 001\}$.}, where the dense layer is $o_j = \sum_{i=1}^n w_{ij}x_i + b_j$. 
Given a known $y$, $\mathbf{x}$ can be inverted successfully, and any additional hidden layers can be inverted by back-propagation. 

Consider a typical classification architecture that uses softmax activation function, $p_k = \frac{e^{o_k}}{\sum_j e^{o_j}}$ to create the model label $y(\mathbf{w}, \mathbf{b}) = c_{\text{arg max}(p_k)}$ (here $k$ is the index for label class), followed by the cross entropy to obtain the cost value, 
\begin{equation}
    \mathcal L(p, y) = -\sum_{k}^{|\mathcal{C}|} y_k \log p_k
\end{equation}
where $|\mathcal{C}|$ is the number of output classes. The derivative of $p_k$ with respect to each $o_j$ is,
\begin{equation}
\frac{\partial p_k}{\partial o_j} = 
\begin{cases}
& p_k(1 - p_j), \hspace{2mm} k=j \\ 
 & -p_kp_j, \hspace{2mm} k\neq j 
\end{cases}
\end{equation}

Now we can calculate the derivative of the cost function with respect to the weights and biases via backpropagation,
\begin{align}
    & \frac{\partial \mathcal L}{\partial w_{i, j=k}}
    = (p_j - y_j)x_i, \\
    &\frac{\partial \mathcal L}{\partial b_{j=k}} =p_j - y_j. 
\end{align}
From this, it can be seen that the number of gradient equations shared with the server would be $n|\mathcal{C}| + |\mathcal{C}|$ while the number of unknowns is $n + |\mathcal{C}|$. For example, from the above equations, it can be seen see that $x_i$ can be found from any $j$, using,
\begin{equation}
   \frac{\partial \mathcal L}{\partial w_{i,j=k}}/\frac{\partial \mathcal L}{\partial b_{j=k}} = x_i 
\end{equation}
In the above setting we saw that for batch size $B=1$, the number of unknowns is less than the number of equations and thus the unknowns can be trivially recovered from the system of equations generated by the dense linear layer of the neural network. 

It turns out there is a feasible attack even if we consider the mini-batch size training with $B >1$. Here, the client trains with the inputs $\text{Samp} := [(\mathbf{x}_\kappa, y_\kappa)_{\kappa \in \text{Samp}}]$ and only shares the averaged gradient information (over the data points with $B = |\text{Samp}|$) with the server,
\begin{align}
        & \frac{\partial \mathcal L}{\partial w_{i, j=k}} = \frac{1}{B} \sum_{\kappa \in \text{Samp}}(p_{\kappa j} - y_{\kappa j})x_{\kappa i}, \\
        &\frac{\partial \mathcal L}{\partial b_{j=k}} = \frac{1}{B} \sum_{\kappa \in \text{Samp}}p_{\kappa j} - y_{\kappa j} 
\end{align}
In this scenario, the number of equations shared is still $n|\mathcal{C}| + |\mathcal{C}|$, whereas the number of unknowns is now $B(n + |\mathcal{C}|)$. Thus the number of unknowns can now exceed the number of equations, resulting in no unique solution for the server when attempting to solve the system of equations. Even in the case that a unique solution exists, numerical optimization can be challenging. 
However, 
In cases where the softmax follows the cross-entropy, the authors in ~\cite{eloul2022enhancing} demonstrate that an accurate direct solution can be achieved in many instances, even when $B \gg 1$. This is attributed to the demixing property across the batch, facilitating the server's ease in retrieving the data points in \text{Samp}.

\section{Protection mechanisms in classical federated learning}\label{app:classical_protection}

To safeguard clients' data from diverse adversaries, including gradient inversion attacks, classical federated learning protocols have implemented multiple protection mechanisms. In this context, we provide a concise review of these techniques.

Homomorphic Encryption (HE)~\cite{rivest1978data,aono2017privacy} stands out as a widely employed encryption technique for privacy protection, enabling the aggregation of local gradient information directly on encrypted data without the need for decryption. However, the practical application of HE faces challenges due to the substantial computational and communication overhead it introduces, particularly for large models~\cite{10.5555/3489146.3489179}.

Another prevalent technique involves the application of differential privacy~\cite{huang2021evaluating,Abadi2016DP,Dwork2006,dwork2006differential,truex2020ldpfed}, wherein Laplace noise or Gaussian noise is typically added to gradient information. While this method is straightforward and minimally impacts communication and computation costs, it may compromise privacy and result in diminished model utility.

Additionally, secret sharing techniques~\cite{Shamir1979,bonawitz2016practical} have been developed to distribute a secret among a group of participants, such as sharing local gradient information with the server using secret sharing among clients. Nevertheless, this approach necessitates extensive message exchange, incurring a communication overhead that may be impractical in many federated learning settings. A concrete example of such techniques is presented in Sec.~\ref{sec:baseline} of the main text.

\end{document}